\documentclass[12pt,technote, onecolumn, draft]{IEEEtran}
\usepackage{latexsym}
\usepackage{mathrsfs}
\usepackage{multirow}
\usepackage{amsfonts,amssymb,amsmath,amsthm,bm}
\usepackage{color}
\usepackage{cite}
\usepackage{comment}
\newtheorem{theorem}{Theorem}
\newtheorem{example}{Example}
\newtheorem{definition}{Definition}
\newtheorem{remark}{Remark}
\newtheorem{corollary}{Corollary}

\newtheorem{proposition}{Proposition}

\begin{document}

\title{Characterizations and MacWilliams-Like Identities of Weakly Regular Bent Partitions $^{\dag}$}
\author{Jiaxin Wang, Yadi Wei, Fang-Wei Fu, Jin Li, Fulin Li
\IEEEcompsocitemizethanks{\IEEEcompsocthanksitem Jiaxin Wang, Jin Li and Fulin Li are with the School of Mathematics, Hefei University of Technology, Hefei 230601, China, Emails: wjiaxin@hfut.edu.cn, lijin\_0102@126.com, lflsxx66@163.com, Yadi Wei and Fang-Wei Fu are with the Chern Institute of Mathematics and LPMC, Nankai University, Tianjin 300071, China, Emails: wydecho@mail.nankai.edu.cn, fwfu@nankai.edu.cn.
}
\thanks{$^\dag$Jiaxin Wang is supported by the National Natural Science Foundation of China (Grant Nos.126011028 and 12671649), Fang-Wei Fu is supported by the National Key Research and Development Program of China (Grant No.2022YFA1005000), the National Natural Science Foundation of China (Grant No. 62371259), the Fundamental Research Funds for the Central Universities of China (Nankai University), and the Nankai Zhide Foundation, Jin Li is supported by the National Natural Science Foundation of China (Grant No.12671649), and Fulin Li is supported by the National Natural Science Foundation of China (Grant No.12471491)}
\thanks{manuscript submitted October 30, 2026}
}

\maketitle

\begin{abstract}
Bent partitions play a significant role in the theory of bent functions. In this paper, we study weakly regular bent partitions, that is, all generated bent functions are weakly regular  (called in class $\mathcal{WRBP}$). Let $\Gamma=\{A_{i}, i \in V_{m}\}$ be a bent partition of $V_{n}$ in class $\mathcal{WRBP}$, where $V_{n}$ denotes an $n$-dimensional vector space over the prime field $\mathbb{F}_{p}$, 
$n$ is even, $m \leq \frac{n}{2}$, and $m \geq 2$ if $p=2$. We first characterize bent partition $\Gamma$ in terms of the duals of component functions of vectorial bent functions. In particular, we prove that if $p \neq 3$, then all bent functions generated by $\Gamma$ are regular. Building on the above result, we give characterizations of bent partition $\Gamma$ in terms of two classes of linear codes. When $n \geq 4$ and $m=\frac{n}{2}$, except for the case where $(p, n)=(2, 4)$ and $|A_{i_{0}}|=1$ for some $i_{0}$, we show that $\Gamma$ is obtained from a spread (after a shift). Particularly, any normal bent partition of $V_{n}$ of depth $p^{\frac{n}{2}}$ in class $\mathcal{WRBP}$ must be obtained from a spread (after a shift), where $n \geq 4$. We establish the Macwilliams-like identities for bent partitions in class $\mathcal{WRBP}$. Furthermore, if $p$ is odd, we present some nonexistence results on bent partitions in class $\mathcal{WRBP}$.
\end{abstract}


\section{Introduction}
Bent functions introduced in \cite{R1976On} have been extensively studied due to their significance in cryptography, coding theory, sequences and combinatorics. In \cite{AM2022Be}, Anbar and Meidl proposed the concept of bent partitions, which play an important role in constructing bent functions. Furthermore, bent partitions have applications in constructing partial difference sets, association schemes and LP-packings (see e.g. \cite{AAKM2024Be,AAKM2025Be,AKM2022Be,AKM2024Am,WFW2023Be}). 

Let $V_{n}$ be an $n$-dimensional vector space over the prime field $\mathbb{F}_{p}$, where $n$ is even. A partition $\Gamma=\{A_{i}, 1 \leq i \leq K\}$ of $V_{n}$ (where $p \mid K$) is said to be a bent partition, if any $p$-ary function $f$ that takes every value $j \in \mathbb{F}_{p}$ on exactly $\frac{K}{p}$ of sets $A_{i}$ is a bent function. If all generated bent functions are weakly regular, then the bent partition $\Gamma$ is called weakly regular (also called in class $\mathcal{WRBP}$). By the results in \cite{AM2022Be} and \cite{WWF2026Fu}, the depth $K$ of bent partitions of $V_{n}$ in class $\mathcal{WRBP}$ satisfies $K=p^m$, where $m \leq \frac{n}{2}$. Up to now, all known bent partitions are in class $\mathcal{WRBP}$. There is a sufficient condition to construct bent partitions in class $\mathcal{WRBP}$ given in \cite{WWF2026Fu}. Specifically, the preimage set partition of a vectorial bent function $F: V_{n} \rightarrow V_{m}$ is a bent partition in class $\mathcal{WRBP}$, if the components $F_{c}, c \in V_{m} \backslash \{0\}$, are all regular or all weakly regular but not regular, and the duals $(F_{c})^{*}, c \in V_{m} \backslash \{0\}$, satisfy $(F_{c})^{*}(x)=G_{c}(x)+h(x)$ for some functions $G: V_{n} \rightarrow V_{m}$ and $h: V_{n} \rightarrow \mathbb{F}_{p}$. When $p=2$, \cite{WFWY2024A} proved that the sufficient condition given in \cite{WWF2026Fu} is also necessary. In the survey of bent partitions  \cite{AKM2026A}, some questions are listed, two of which are as follows:

\textbf{Question 1}: Prove or disprove that a (normal) bent partition $\Gamma=\{A_{i}, i \in V_{n/2}\}$ of $V_{n}$ of depth $p^{\frac{n}{2}}$ must be from a spread.

\textbf{Question 2}: When $p$ is odd, determine whether the sufficient condition of constructing bent partitions in class $\mathcal{WRBP}$ given in \cite{WWF2026Fu} (introduced as above) is also necessary.

For Question 1, when $n \geq 4, p \neq 2$, under the condition that $A_{i}=-A_{i}$ for all $i \in V_{n/2}$, and if $p=3$, then $A_{i} \neq \{0\}$ for all $i \in V_{n/2}$, \cite{WWFLL2026Co} showed that the bent partition $\Gamma$ must be obtained from a spread. In this paper, we consider the case that $\Gamma$ is in class $\mathcal{WRBP}$. Note that there exist bent partitions in class $\mathcal{WRBP}$ with $A_{i} \neq -A_{i}$ for some $i$ 
(see \cite{AFKMWW2026A,WWF2026Fu}). 

As an important result in coding theory, MacWilliams identities state that the weight enumerator of a linear code uniquely determines the weight enumerator of its dual code. In \cite{HL2016Ex,HLL2012Ma,HLL2014Bo,HLL2014Ne}, MacWilliams-like identities for $p$-ary bent functions are provided. Based on these results, some nonexistence results on $p$-ary bent functions are presented. 

In this paper, we study the characterizations and MacWilliams-like identities of bent partitions in class $\mathcal{WRBP}$. We summarize our main contributions as follows. 
\begin{itemize}
 \item For Question 2, we give a positive answer. In particular, we show that if $p \neq 3$, then all bent functions generated by the bent partition in class $\mathcal{WRBP}$ are regular. 
 \item We give characterizations of bent partitions in class $\mathcal{WRBP}$ in terms of two classes of linear codes.
 \item Let $n \geq 4$. For Question 1, under the cases that $p$ is odd, or $p=2, n \geq 6$, or $(p, n)=(2, 4)$ and $|A_{i}|
 \neq 1$ for all $i \in V_{n/2}$, we prove that the bent partition $\Gamma$ of $V_{n}$ of depth $p^{\frac{n}{2}}$ in class $\mathcal{WRBP}$ is obtained from a spread (after a shift). Particularly, any normal bent partition of $V_{n}$ of depth $p^{\frac{n}{2}}$ in class $\mathcal{WRBP}$ is obtained from a spread (after a shift). 
 \item We provide the Macwillams-like identities for bent partitions in class $\mathcal{WRBP}$. Furthermore, some nonexistence results on bent partitions in class $\mathcal{WRBP}$ are presented.
\end{itemize}

The rest of the paper is organized as follows. Section II gives the needed preliminaries. In Section III, we characterize bent partitions in class $\mathcal{WRBP}$ via the duals of components of vectorial bent functions and two classes of linear codes. We study the correspondence between bent partitions in class $\mathcal{WRBP}$ whose depth achieves the maximal value and spreads. In Section IV, we establish the MacWilliams-like identities for bent partitions in class $\mathcal{WRBP}$, based on which some nonexistence results are given. In Section V, a conclusion is given.

\section{Preliminaries}
In this section, we present some needed results. First, we give some fixed notations unless otherwise stated. 
\begin{itemize}
\item $p$ is a prime.
\item $\zeta_{p}$ is a complex primitive $p$-th root of unity.
\item $\mathbb{F}_{p^n}$ is the finite field with $p^n$ elements.
\item $\mathbb{F}_{p}^{n}$ is the vector space of the $n$-tuples over $\mathbb{F}_{p}$.
\item $V_{n}$ denotes an $n$-dimensional vector space over $\mathbb{F}_{p}$.
\item $\langle , \rangle_{n}$ denotes a (non-degenerate) inner product of $V_{n}$. When $V_{n}=\mathbb{F}_{p^n}$, let $\langle x, y \rangle_{n}=\mathrm{Tr}_{1}^{n}(xy)$, where $\mathrm{Tr}_{1}^{n}$ is the trace function from $\mathbb{F}_{p^n}$ to $\mathbb{F}_{p}$; when $V_{n}=\mathbb{F}_{p}^{n}$, let $\langle x, y \rangle_{n}=x \cdot y$, where $\cdot$ denotes the usual dot product.
\item For a function $F: V_{n} \rightarrow V_{m}$, let $D_{F, i}=\{x \in V_{n}: F(x)=i\}$, where $i \in V_{m}$.
\item For any set $D \subseteq V_{n}$ and $a \in V_{n}$, let $\chi_{a}(D)=\sum_{x \in D}\chi_{a}(x)$, where $\chi_{a}$ is the character defined by $\chi_{a}(x)=\zeta_{p}^{\langle a, x \rangle_{n}}$.
\item For a set $D \subseteq V_{n}$, let $D^{*}=\{x \in D: x \neq 0\}$, and
let $\delta_{D}$ be the indicator function. If $D=\{\alpha\}$, let $\delta_{\alpha}=\delta_{\{\alpha\}}$.
\item For integers $i, j$ with $i>j$, define $\binom{j}{i}=0$.
\end{itemize}
\subsection{Bent partitions}
The \emph{Walsh transform} of a $p$-ary function $f: V_{n} \rightarrow \mathbb{F}_{p}$ is defined by 
\begin{equation*}
	W_{f}(a)=\sum_{x \in V_{n}}\zeta_{p}^{f(x)-\langle a, x \rangle_{n}}, a \in V_{n}.
\end{equation*}
If $|W_{f}(a)|=p^{\frac{n}{2}}$ for each $a \in V_{n}$, then $f$ is called a \emph{bent} function. For a function $F: V_{n} \rightarrow V_{m}$, define $F_{c}: V_{n} \rightarrow \mathbb{F}_{p}$ ($c \in V_{m}^{*}$) as $F_{c}(x)=\langle c, F(x)\rangle_{m}$, which is called a \emph{component function} of $F$. If all component functions $F_{c}, c \in V_{m}^{*}$, are bent functions, then $F$ is called a \emph{vectorial bent} function. The Walsh transform of a bent function $f: V_{n} \rightarrow \mathbb{F}_{p}$ satisfies
 \begin{equation*}
 W_{f}(a)=\epsilon_{f, a} p^{\frac{n}{2}}\zeta_{p}^{f^{*}(a)}, a \in V_{n},
 \end{equation*}
where $\epsilon_{f, a}=1$ if $p=2$,  $\epsilon_{f, a}=\pm 1$ if $p^{n} \equiv 1 \pmod 4$, $\epsilon_{f, a}=\pm \sqrt{-1}$ if $p^{n} \equiv 3 \pmod 4$, and $f^{*}: V_{n} \rightarrow \mathbb{F}_{p}$ is called the \emph{dual} of $f$. If $\epsilon_{f, a}=\epsilon_{f}$ for all $a \in V_{n}$, where $\epsilon_{f} \in \{\pm 1, \pm \sqrt{-1}\}$, then $f$ is called \emph{weakly regular}. In particular, if $\epsilon_{f}=1$, then $f$ is called \emph{regular}.

Anbar and Meidl \cite{AM2022Be} introduced the notion of bent partitions as follows.

\begin{definition}
Let $n$ be an even positive integer, and $K$ be a positive integer with $p \mid K$. 
\begin{itemize}
\item Let $\Gamma=\{A_{i}, 1 \leq i \leq K\}$ be a partition of $V_{n}$. If each $p$-ary function $f: V_{n} \rightarrow \mathbb{F}_{p}$ for which any $j \in \mathbb{F}_{p}$ has exactly $\frac{K}{p}$ of sets $A_{i}$ in $\Gamma$ in its preimage set, is a bent function, then $\Gamma$ is called a bent partition of $V_{n}$ of depth $K$. Every such bent function $f$ is called a bent function generated by bent partition $\Gamma$. 
\item Let $\Omega=\{U, A_{i}, 1 \leq i\leq K\}$ be a partition of $V_{n}$, where $U$ is an $\frac{n}{2}$-dimensional affine subspace of $V_{n}$. Then $\Omega$ is called a normal bent partition of $V_{n}$ of depth $K$ if each $p$-ary function $f: V_{n} \rightarrow \mathbb{F}_{p}$ satisfying the following conditions is bent:

(1) Any $j \in \mathbb{F}_{p}$ has exactly $\frac{K}{p}$ of sets $A_{i}$ in $\Gamma$ in its preimage set.

(2) $f$ is constant on $U$.
\end{itemize}
\end{definition}
If $\Omega=\{U, A_{i}, 1 \leq i \leq K\}$ is a normal bent partition, then for arbitrary $1 \leq i_{0} \leq K$, $\Gamma=\{U \cup A_{i_{0}}, A_{i}, i \in \{1, 2, \dots, K\} \backslash \{i_{0}\}\}$ is a bent partition. We say that $\Gamma$ is a \emph{bent partition arising from the normal bent partition $\Omega$}.  For a (normal) bent partition $\Gamma$, if all generated bent functions $f$ are weakly regular, then we call $\Gamma$ a \emph{weakly regular (normal) bent partition} (simply called in class $\mathcal{WRBP}$). By \cite[Theorem 3]{AM2022Be} and \cite[Theorem 3]{WWF2026Fu}, the depth $K$ of a (normal) bent parition of $V_{n}$ in class $\mathcal{WRBP}$ satisfies $K=p^{m}$, where $m \leq \frac{n}{2}$. By \cite[Remark 1]{WWF2026Fu}, for a bent partition in class $\mathcal{WRBP}$, the generated bent functions $f$ are all regular or all weakly regular but not regular, that is, $\epsilon_{f}=\epsilon$, where $\epsilon \in \{\pm 1\}$ is a constant. By \cite[Corollary 1]{AM2022Be}, all generated bent functions by a normal bent partition in class $\mathcal{WRBP}$ are regular.
When $p=2$, each bent partition is in class $\mathcal{WRBP}$. When $p$ is odd, for a bent partition  $\Gamma=\{A_{i}, i \in V_{m}\}$ of $V_{n}$ in class $\mathcal{WRBP}$, by \cite[Theorem 3]{AM2022Be}, there is $k \in V_{m}$ such that
\begin{equation}\label{add}
	|A_{k}|=p^{n-m}+\epsilon p^{\frac{n}{2}-m} (p^m-1), |A_{i}|=p^{n-m}-\epsilon p^{\frac{n}{2}-m}, i \neq k.
\end{equation}

Suppose $U_{i}, 0 \leq i \leq t-1$, are $m$-dimensional linear subspaces of $V_{2m}$ satisfying $U_{i} \cap U_{j}=\{0\}$ for any $i \neq j$. Then $\{U_{i}, 0 \leq i \leq t-1\}$ is called a \emph{partial spread} of $V_{2m}$. If $t=p^{m}+1$, then $\{U_{i}, 0 \leq i \leq p^{m}\}$ is called a \emph{spread} of $V_{2m}$. For a spread $\{U_{i}, 0 \leq i \leq p^{m}\}$ of $V_{2m}$, let $A_{1}=(U_{0} \cup U_{1})-\alpha, A_{i}=(U_{i}-\alpha)\backslash \{-\alpha\}, 2 \leq i \leq p^m$, where $\alpha \in V_{n}$. Then $\{A_{i}, 1 \leq i \leq p^{m}\}$ is a bent partition of $V_{2m}$ in class $\mathcal{WRBP}$ whose depth achieves the maximal value.
\subsection{Linear codes}
For any $x=(x_{1}, \dots, x_{n}) \in \mathbb{F}_{p}^{n}$, let $supp(x)=\{i: x_{i} \neq 0, 1 \leq i \leq n\}$ be the support of $x$, and let  $wt(x)=|supp(x)|$ be the Hamming weight of $x$. For $x, y \in \mathbb{F}_{p}^{n}$, let $d(x, y)=wt(x-y)$ be the Hamming distance between $x$ and $y$. A $p$-ary $[n, k, d]$ \emph{linear code} $\mathcal{C}$ is a $k$-dimensional linear subspace of $\mathbb{F}_{p}^{n}$ with minimum Hamming distance $d$, where $d=\mathrm{min}\{d(x, y): x, y \in \mathcal{C}, x \neq y\}$. Let $M_{j}$ be the
number of codewords with Hamming weight $j$ in $\mathcal{C}$ for $1 \leq j \leq n$. Then $1+\sum_{j=1}^{n}M_{j}z^{j}$ is called the \emph{weight enumerator} of $\mathcal{C}$. If the number of nonzero $M_{j}, 1 \leq j \leq n$, is $t$, then the code $\mathcal{C}$ is called \emph{$t$-weight}.
\subsection{Hadamard matrices}
An $n \times n$ complex matrix $H$ is called a \emph{Buston Hadamard matrix} if every entry of $H$ is an $m$-th root of unity and $HH^{*}=nI_{n}$, where $H^{*}$ denotes the conjugate transpose of $H$, and $I_{n}$ is the $n \times n$ identity matrix. By \cite[Property 4]{KSW1985Ge}, $f: V_{n} \rightarrow \mathbb{F}_{p}$ is bent if and only if $H=[\zeta_{p}^{f(x-y)}]_{x, y \in V_{n}}$ is a Buston Hadamard matrix. If $f: V_{n} \rightarrow \mathbb{F}_{p}$ is weakly regular bent, then we say that the Buston Hadamard matrix $H$ is of \emph{weakly regular type}, and denote $\vartheta_{H}=\epsilon_{f}$.

\section{Characterization of bent partitions in class $\mathcal{WRBP}$}
In this section, we will characterize bent partitions in class $\mathcal{WRBP}$ via the duals of  component functions of vectorial bent functions, Buston Hadamard matrices and two classes of linear codes. Furthermore, we will study the correspondence between spreads and bent partitions in class $\mathcal{WRBP}$ whose depth achieves the maximal value.

The following theorem gives the characterization of bent partitions in class $\mathcal{WRBP}$ in terms of the duals of component functions of vectorial bent functions, which solves Question 2.

\begin{theorem} \label{Theorem 1}
Let $\Gamma=\{A_{i}, i \in V_{m}\}$ be a partition of $V_{n}$, where $n$ is even, $m \leq \frac{n}{2}$, and $m \geq 2$ if $p=2$. Define $F(x)=i, \text{ if } x \in A_{i}, i \in V_{m}$. Then the following statements are equivalent.

(i) $\Gamma$ is a bent partition in class $\mathcal{WRBP}$.

(ii) $F$ is a vectorial bent function for which all component functions $F_{c}, c \in V_{m}^{*}$, are weakly regular with $\epsilon_{F_{c}}=\epsilon$, where $\epsilon \in \{\pm 1\}$ with $\epsilon=1$ if $p \neq 3$, and there exist functions $G: V_{n} \rightarrow V_{m}$ and $h: V_{n} \rightarrow \mathbb{F}_{p}$ such that $(F_{c})^{*}=G_{c}+h$ for all $c \in V_{m}^{*}$.

(iii) $F$ is a vectorial bent function for which all component functions $F_{c}, c \in V_{m}^{*}$, are weakly regular with $\epsilon_{F_{c}}=\epsilon$, where $\epsilon \in \{\pm 1\}$ with $\epsilon=1$ if $p \neq 3$, and $(F_{d})^{*}+(F_{e})^{*}-(F_{d+e})^{*}, d \neq -e \in V_{m}^{*}$, are the same function.
\end{theorem}

\begin{proof}
By \cite[Theorem 4]{WWF2026Fu}, $(ii) \Rightarrow (i)$ holds. By the proof of \cite[Proposition 2]{WWF2026Fu}, $(ii) \Leftrightarrow (iii)$ holds. We only need to show that $(i) \Rightarrow (ii)$.

$(i) \Rightarrow (ii)$: We first show that $F$ is a vectorial bent function whose component functions are all regular or all weakly regular but not regular, and there exist functions $G: V_{n} \rightarrow V_{m}$ and $h: V_{n} \rightarrow \mathbb{F}_{p}$ such that $(F_{c})^{*}=G_{c}+h$ for all $c \in V_{m}^{*}$. By \cite[Theorem 2]{WWF2026Fu} and \cite[Proposition 2]{WFWY2024A}, we only need to consider the case of odd $p$ and $m \geq 2$. In the following, let $p$ be odd and $m \geq 2$, and let $a$ be an arbitrary fixed nonzero element in $V_{n}$. Define the equivalence relationship $R_{a}$ in $\Gamma$ by $A_{i} R_{a} A_{j}$ if and only if $\chi_{a}(A_{i})=\chi_{a}(A_{j})$. Let $\{K_{1}^{(a)}, \dots, K_{s}^{(a)}\}$ be the corresponding equivalence classes. We claim that $s=2$. To the contrary, assume that $s \geq 3$. Let $f$ be an arbitrary $p$-ary bent function generated by $\Gamma$. Then $\{D_{f, i}, i \in \mathbb{F}_{p}\}$ is a bent partition in class $\mathcal{WRBP}$. For any $j \in \mathbb{F}_{p}$, by \cite[Theorem 2]{WWF2026Fu} and its proof,
\begin{equation}\label{th11}
\chi_{a}(D_{f, j})=\begin{cases}
\epsilon p^{\frac{n}{2}-1}(p-1)\zeta_{p}^{h_{f}(-a)}, & \text{ if } g_{f}(-a)=j,\\
-\epsilon p^{\frac{n}{2}-1}\zeta_{p}^{h_{f}(-a)}, & \text{ if } g_{f}(-a) \neq j,
\end{cases}
\end{equation}
where $g_{f}(x)=\frac{f^{*}(x)+f^{*}(-x)}{2}$, $h_{f}(x)=\frac{f^{*}(x)-f^{*}(-x)}{2}$, and $\epsilon \in \{\pm 1\}$ is a constant. By Eq. \eqref{th11}, we have $s \geq 2$. Let $A_{i_{1}}, A_{i_{2}}$ be arbitrary sets with $\chi_{a}(A_{i_{1}}) \neq \chi_{a}(A_{i_{2}})$. Let $f'$ be a $p$-ary bent function generated by $\Gamma$ with $A_{i_{1}} \subseteq D_{f', 0}, A_{i_{2}} \subseteq D_{f', 1}$. Let $f''$ be the $p$-ary bent function generated by $\Gamma$ for which $D_{f'', 0}=A_{i_{2}} \cup D_{f', 0} \backslash A_{i_{1}}, D_{f'', 1}=A_{i_{1}} \cup D_{f', 1} \backslash A_{i_{2}}, D_{f'', i}=D_{f', i}, i \in \mathbb{F}_{p} \backslash \{0, 1\}$. By Eq. \eqref{th11} and $\chi_{a}(D_{f', i})=\chi_{a}(D_{f'', i}), i \in \mathbb{F}_{p} \backslash \{0, 1\}$, we have $h_{f'}(-a)=h_{f''}(-a)$ (otherwise, $(p-1)\zeta_{p}^{h_{f'}(-a)}+\zeta_{p}^{h_{f''}(-a)}=0$ or $(p-1)\zeta_{p}^{h_{f''}(-a)}+\zeta_{p}^{h_{f'}(-a)}=0$, which is impossible). Since $0 \neq \chi_{a}(A_{i_{1}})-\chi_{a}(A_{i_{2}})=\chi_{a}(D_{f', 0})-\chi_{a}(D_{f'', 0})$, then by $h_{f'}(-a)=h_{f''}(-a)$ and Eq. \eqref{th11}, we have  $\chi_{a}(A_{i_{1}})-\chi_{a}(A_{i_{2}})=\pm p^{\frac{n}{2}}\zeta_{p}^{h_{f'}(-a)}$. By the above arguments, if $A_{i_{1}}$ and $A_{i_{2}}$ are in different equivalence classes, then
\begin{equation}\label{th12}
\chi_{a}(A_{i_{1}})-\chi_{a}(A_{i_{2}})=\varepsilon_{a, i_{1}, i_{2}} p^{\frac{n}{2}}\zeta_{p}^{\tau_{a, i_{1}, i_{2}}},
\end{equation}
where $\varepsilon_{a, i_{1}, i_{2}} \in \{\pm 1\}$ and $\tau_{a, i_{1}, i_{2}} \in \mathbb{F}_{p}$. Let $A_{i_{1}}, A_{i_{2}}, A_{i_{3}}$ be arbitrary sets from different equivalence classes. Then
\begin{align*}
&\chi_{a}(A_{i_{1}})-\chi_{a}(A_{i_{2}})=\varepsilon_{a, i_{1}, i_{2}} p^{\frac{n}{2}}\zeta_{p}^{\tau_{a, i_{1}, i_{2}}},\\ &\chi_{a}(A_{i_{3}})-\chi_{a}(A_{i_{1}})=\varepsilon_{a, i_{3}, i_{1}} p^{\frac{n}{2}}\zeta_{p}^{\tau_{a, i_{3}, i_{1}}},\\ &\chi_{a}(A_{i_{2}})-\chi_{a}(A_{i_{3}})=\varepsilon_{a, i_{2}, i_{3}} p^{\frac{n}{2}}\zeta_{p}^{\tau_{a, i_{2}, i_{3}}},
\end{align*}
where $\varepsilon_{a, i_{1}, i_{2}}, \varepsilon_{a, i_{3}, i_{1}}, \varepsilon_{a, i_{2}, i_{3}} \in \{\pm 1\}$ and $\tau_{a, i_{1}, i_{2}}$, $\tau_{a, i_{3}, i_{1}}, \tau_{a, i_{2}, i_{3}} \in \mathbb{F}_{p}$. From which,
\begin{equation}\label{th13}
	\varepsilon_{a, i_{1}, i_{2}}\zeta_{p}^{\tau_{a}, i_{1}, i_{2}}+\varepsilon_{a, i_{2}, i_{3}}\zeta_{p}^{\tau_{a}, i_{2}, i_{3}}+\varepsilon_{a, i_{3}, i_{1}}\zeta_{p}^{\tau_{a}, i_{3}, i_{1}}=0.
\end{equation}
\begin{itemize}
	\item When $p \geq 5$, Eq. \eqref{th13} contradicts the fact that $1+x+x^{2}+\dots+x^{p-1}$ is the minimal polynomial of $\zeta_{p}$ over $\mathbb{Q}$. Thus, $s=2$ when $p \geq 5$.
	\item When $p=3$, by Eq. \eqref{th13} and the fact that $1+x+x^{2}$ is the minimal polynomial of $\zeta_{3}$ over $\mathbb{Q}$, we have
	\begin{equation*}
	\varepsilon_{a, i_{1}, i_{2}}=\varepsilon_{a, i_{2}, i_{3}}=\varepsilon_{a, i_{3}, i_{1}}, \{\tau_{a, i_{1}, i_{2}}, \tau_{a, i_{2}, i_{3}}, \tau_{a, i_{3}, i_{1}}\}=\mathbb{F}_{3}.
	\end{equation*}
If $s \geq 4$, let $A_{i_{1}}, A_{i_{2}}, A_{i_{3}}, A_{i_{4}}$  be in different equivalence classes. Then $\varepsilon_{a, i_{1}, i_{2}}=\varepsilon_{a, i_{2}, i_{4}}=\varepsilon_{a, i_{4}, i_{1}}(=\varepsilon_{a, i_{3}, i_{1}})$ and $\{\tau_{a, i_{1}, i_{2}}, \tau_{a, i_{2}, i_{4}}, \tau_{a, i_{4}, i_{1}}\}=\mathbb{F}_{3}$. As $\chi_{a}(A_{i_{3}}) \neq \chi_{a}(A_{i_{4}})$, then $\tau_{a, i_{4}, i_{1}}=\tau_{a, i_{2}, i_{3}}, \tau_{a, i_{2}, i_{4}}=\tau_{a, i_{3}, i_{1}}$. Thus,
\begin{equation*}
\chi_{a}(A_{i_{4}})-\chi_{a}(A_{i_{3}})=
	\varepsilon_{a, i_{4}, i_{1}}3^{\frac{n}{2}}
	\zeta_{3}^{\tau_{a, i_{4}, i_{1}}}-\varepsilon_{a, i_{3}, i_{1}}3^{\frac{n}{2}}\zeta_{3}^{\tau_{a, i_{3}, i_{1}}}=\varepsilon_{a, i_{4}, i_{1}}3^{\frac{n}{2}}
	(\zeta_{3}^{\tau_{a, i_{4}, i_{1}}}-\zeta_{3}^{\tau_{a, i_{2}, i_{4}}}),
\end{equation*}
 which contradicts Eq. \eqref{th12}. Hence, $s=3$. W.l.o.g., suppose $|K^{(a)}_{1}| \geq |K^{(a)}_{2}| \geq |K^{(a)}_{3}|$ (Note that $|K^{(a)}_{1}| \geq 3$ as $3^m \geq 9$) and $A_{i_{0}}, A_{i_{1}} \in K^{(a)}_{1}, A_{i_{2}} \in K^{(a)}_{2}, A_{i_{3}} \in K^{(a)}_{3}$. Let $f'$ be a ternary bent function generated by $\Gamma$ with $A_{i_{0}} \cup A_{i_{1}} \subseteq D_{f', 0}$ and $A_{i_{2}} \cup A_{i_{3}} \subseteq D_{f', 1}$. Let $f''$ be the ternary bent function for which $D_{f'', 0}=A_{i_{2}} \cup A_{i_{3}} \cup D_{f', 0} \backslash (A_{i_{0}} \cup A_{i_{1}}), D_{f'', 1}=A_{i_{0}} \cup A_{i_{1}} \cup D_{f', 1} \backslash (A_{i_{2}} \cup A_{i_{3}})$ and $D_{f'', 2}=D_{f', 2}$. By Eq. \eqref{th11} and $\chi_{a}(D_{f', 2})=\chi_{a}(D_{f'', 2})$, we have  $h_{f'}(-a)=h_{f''}(-a)$ and
 \begin{equation*}
 	\begin{split}
 0 \text{ or } \pm 3^{\frac{n}{2}}\zeta_{3}^{h_{f'}(-a)}=\chi_{a}(D_{f', 0})-\chi_{a}(D_{f'', 0})&=(\chi_{a}(A_{i_{1}})-\chi_{a}(A_{i_{2}}))-(\chi_{a}(A_{i_{3}})-\chi_{a}(A_{i_{1}}))\\
 &=\varepsilon_{a, i_{1}, i_{2}}3^{\frac{n}{2}}\zeta_{3}^{\tau_{a, i_{1}, i_{2}}}-\varepsilon_{a, i_{3}, i_{1}}3^{\frac{n}{2}}\zeta_{3}^{\tau_{a, i_{3}, i_{1}}},
 \end{split}
\end{equation*}
which is impossible since $\varepsilon_{a, i_{1}, i_{2}}=\varepsilon_{a, i_{3}, i_{1}}$ and $\tau_{a, i_{1}, i_{2}} \neq \tau_{a, i_{3}, i_{1}}$. Therefore, $s=2$ when $p=3$.
\end{itemize}
W.l.o.g., suppose $|K^{(a)}_{1}| \geq |K^{(a)}_{2}|$. We show that $|K^{(a)}_{1}|=p^m-1$ and $|K^{(a)}_{2}|=1$. To the contrary, assume that $|K^{(a)}_{2}| \geq 2$. Let $A_{i_{1}}, A_{i_{2}} \in K^{(a)}_{1}, A_{i_{3}}, A_{i_{4}} \in K^{(a)}_{2}$. Let $f'$ be a $p$-ary bent function generated by $\Gamma$ with $A_{i_{1}} \cup A_{i_{2}} \subseteq D_{f', 0}, A_{i_{3}} \cup A_{i_{4}} \subseteq D_{f', 1}$. Let $f''$ be the $p$-ary bent function for which $D_{f'', 0}=A_{i_{3}} \cup A_{i_{4}} \cup D_{f', 0} \backslash (A_{i_{1}} \cup A_{i_{2}}), D_{f'', 1}=A_{i_{1}} \cup A_{i_{2}} \cup D_{f', 1} \backslash (A_{i_{3}} \cup A_{i_{4}})$ and $D_{f'', i}=D_{f', i}, i \in \mathbb{F}_{p} \backslash \{0, 1\}$. By Eq. \eqref{th11} and $\chi_{a}(D_{f', i})=\chi_{a}(D_{f'', i})$ for any $i \in \mathbb{F}_{p} \backslash \{0, 1\}$, we have  $h_{f'}(-a)=h_{f''}(-a)$ and
\begin{equation*}
 0 \text{ or } \pm p^{\frac{n}{2}}\zeta_{p}^{h_{f'}(-a)}=\chi_{a}(D_{f', 0})-\chi_{a}(D_{f'', 0})=2(\chi_{a}(A_{i_{1}})-\chi_{a}(A_{i_{3}})),
 \end{equation*}
 which contradicts $A_{i_{1}} \in K^{(a)}_{1}, A_{i_{3}} \in K^{(a)}_{2}$ and Eq. \eqref{th12}.
By the above arguments, we have proven that for arbitrary fixed $a \in V_{n}^{*}$, $\chi_{a}(A_{i}), i \in V_{m}$, take two values $X$ and $X+\varepsilon_{a} p^{\frac{n}{2}}\zeta_{p}^{\tau_{a}}$, where the frequency of $X$ and $X+\varepsilon_{a} p^{\frac{n}{2}}\zeta_{p}^{\tau_{a}}$ are $p^m-1$ and $1$ respectively, $\varepsilon_{a} \in \{\pm 1\}$ and $\tau_{a} \in \mathbb{F}_{p}$. Let $i_{a} \in V_{m}$ be with $\chi_{a}(A_{i_{a}})=X+\varepsilon_{a}p^{\frac{n}{2}}\zeta_{p}^{\tau_{a}}$.
Since $\sum_{i \in V_{m}}\chi_{a}(A_{i})=\chi_{a}(V_{n})=0$, we get $X=-\varepsilon_{a}p^{\frac{n}{2}-m}\zeta_{p}^{\tau_{a}}$. Since $\Gamma$ is a bent partition in class $\mathcal{WRBP}$, for any balanced function $B$ from $V_{m}$ to $\mathbb{F}_{p}$ (that is, $|D_{B, i}|=p^{m-1}$ for all $i \in \mathbb{F}_{p}$), we have that $B(F(x))$ is bent and
\begin{equation*}
\epsilon p^{\frac{n}{2}}\zeta_{p}^{(B(F))^{*}(-a)}=W_{B(F)}(-a)=\sum_{x \in V_{n}}\zeta_{p}^{B(F(x))+\langle a, x \rangle_{n}}
=\sum_{i \in V_{m}}\chi_{a}(A_{i})\zeta_{p}^{B(i)}=\varepsilon_{a}p^{\frac{n}{2}}\zeta_{p}^{\tau_{a}+B(i_{a})},
\end{equation*}
which implies that
\begin{equation*}
\varepsilon_{a}=\epsilon,  (B(F))^{*}(-a)=\tau_{a}+B(i_{a}).
\end{equation*}
Define $h: V_{n} \rightarrow \mathbb{F}_{p}$ as
\begin{equation*}
h(0)=0, h(x)=\tau_{-x}, x \in V_{n}^{*}.
\end{equation*}
 Define $G: V_{n} \rightarrow V_{m}$ as
\begin{equation*}
G(0)=i_{0}, G(x)=i_{-x}, x \in V_{n}^{*},
\end{equation*}
where $i_{0} \in V_{m}$ is given by $|A_{i_{0}}| =p^{n-m}+\epsilon p^{\frac{n}{2}-m}(p^m-1)$. Note that $|A_{i}|=p^{n-m}-\epsilon p^{\frac{n}{2}-m}$ for $i \neq i_{0}$. Then for any balanced function $B$ from $V_{m}$ to $\mathbb{F}_{p}$, we have $(B(F))^{*}(x)=B(G(x))+h(x), x \in V_{n}$. In particular, when $B(x)=\langle c, x \rangle_{m}$, where $c \in V_{m}^{*}$, the result holds.

Next, we prove that $\epsilon=1$ if $p \neq 3$. Obviously, $\epsilon=1$ if $p=2$. We only need to consider the case of odd $p$. By the above arguments, there exist functions $G: V_{n} \rightarrow V_{m}$ and $h: V_{n} \rightarrow \mathbb{F}_{p}$ with $h(0)=0$ such that for any $a \in V_{n}$ and $i \in V_{m}$,
\begin{equation*}
\chi_{-a}(A_{i})=p^{n-m}\delta_{0}(a)+\epsilon p^{\frac{n}{2}-m}\zeta_{p}^{h(a)}(p^m \delta_{G(a)}(i)-1).
\end{equation*}
For each $w \in \mathbb{F}_{p}^{*}$, let $\sigma_{w}$ be the automorphism of cyclotomic field $\mathbb{Q}(\zeta_{p})$ defined by $\sigma_{w}(\zeta_{p})=\zeta_{p}^{w}$.
Then for any $a \in V_{n}$ and $i \in V_{m}$,
\begin{equation*}
\begin{split}
p^{n-m}\delta_{0}(a)+\epsilon p^{\frac{n}{2}-m}\zeta_{p}^{wh(a)}(p^{m}\delta_{G(a)}(i)-1)&=\sigma_{w}(\chi_{-a}(A_{i}))\\
&=\chi_{-wa}(A_{i})\\
&=p^{n-m}\delta_{0}(a)+\epsilon p^{\frac{n}{2}-m}\zeta_{p}^{h(wa)}(p^{m}\delta_{G(wa)}(i)-1),
\end{split}
\end{equation*}
which implies that
\begin{equation}\label{th14}
G(wa)=G(a), h(wa)=wh(a), w \in \mathbb{F}_{p}^{*}, a \in V_{n}.
\end{equation}
Let $i \in V_{m}$ be fixed. Then we have
\begin{equation}\label{th15}
	\begin{split}
		&\sum_{a \in V_{n}}|\chi_{-a}(A_{i})|^{2}\\
		&=\big(p^{n-m}+\epsilon p^{\frac{n}{2}-m}(p^{m}\delta_{G(0)}(i)-1)\big)^{2}+\sum_{\substack{a \in V_{n}^{*}\\G(a)=i}}(p^m-1)^{2}p^{n-2m}+\sum_{\substack{a \in V_{n}^{*}\\G(a) \neq i}}p^{n-2m}\\
		&=\big(p^{n-m}+\epsilon p^{\frac{n}{2}-m}(p^{m}\delta_{G(0)}(i)-1)\big)^{2}+p^{n-2m}(p^n-1)+p^{n-m}(p^m-2)|D_{G, i}|\\
		& \ \ \ -p^{n-m}(p^m-2)\delta_{G(0)}(i)\\
		&=2p^{n-m}\big(p^{n-m}-\epsilon p^{\frac{n}{2}-m}+\epsilon p^{\frac{n}{2}}\delta_{G(0)}(i)\big)+p^{n-m}(p^m-2)|D_{G, i}|.
	\end{split}
\end{equation}
On the other hand,
\begin{equation}\label{th16}
	\begin{split}
		\sum_{a \in V_{n}}|\chi_{-a}(A_{i})|^{2}
		&=\sum_{a \in V_{n}}\sum_{\alpha \in A_{i}}\zeta_{p}^{\langle-a, \alpha\rangle_{n}} \sum_{\alpha' \in A_{i}}\zeta_{p}^{\langle a, \alpha'\rangle_{n}}
		=\sum_{\alpha, \alpha' \in A_{i}}\sum_{a \in V_{n}}\zeta_{p}^{\langle a, \alpha'-\alpha\rangle_{n}}
		=p^{n}|A_{i}|\\
		&=p^{n}\big(p^{n-m}-\epsilon p^{\frac{n}{2}-m}+\epsilon p^{\frac{n}{2}}\delta_{G(0)}(i)\big).
	\end{split}
\end{equation}
Combining Eq. \eqref{th15} and Eq. \eqref{th16}, we obtain
\begin{equation}\label{co5}
	|D_{G, i}|=p^{n-m}+\epsilon p^{\frac{n}{2}-m}(p^{m}\delta_{G(0)}(i)-1).
\end{equation}
By Eq. \eqref{th14}, for any $i \in V_{m} \backslash \{G(0)\}$, $(p-1) \mid |D_{G, i}|$. Hence, $\epsilon=1$ if $p \neq 3$.
\end{proof}

\begin{remark}\label{Remark 1}
Let $p=2$, $n \geq 4$ be an even integer and $2 \leq m \leq \frac{n}{2}$. Suppose $\{A_{i}, i \in V_{m}\}$ is a bent partition of $V_{n}$. By \cite[Proposition 2]{WFWY2024A} and its proof, 
	there exist functions $G: V_{n} \rightarrow V_{m}$ and $h: V_{n} \rightarrow \mathbb{F}_{2}$ such that
	\begin{equation}\label{R1}
		\chi_{a}(A_{i})=2^{n-m}\delta_{0}(a)+2^{\frac{n}{2}-m}(-1)^{h(a)}(2^m \delta_{G(a)}(i)-1), a \in V_{n}, i \in V_{m}.
	\end{equation}
	With the same arguments as for odd $p$ in the proof of Theorem \ref{Theorem 1}, if  $|A_{i_{0}}|=2^{n-m}-2^{\frac{n}{2}-m}+2^{\frac{n}{2}}, |A_{i}|=2^{n-m}-2^{\frac{n}{2}-m}, i \neq i_{0}, $ for some $i_{0} \in V_{m}$, (that is, $G(0)=i_{0}, h(0)=0$), then we can obtain $|D_{G, i}|=2^{n-m}+2^{\frac{n}{2}-m}(2^{m}\delta_{G(0)}(i)-1)$; if $|A_{i_{0}}|=2^{n-m}+2^{\frac{n}{2}-m}-2^{\frac{n}{2}}, |A_{i}|=2^{n-m}+2^{\frac{n}{2}-m}, i \neq i_{0}, $ for some $i_{0} \in V_{m}$ (that is, $G(0)=i_{0}, h(0)=1$), then we can obtain $|D_{G, i}|=2^{n-m}-2^{\frac{n}{2}-m}(2^{m}\delta_{G(0)}(i)-1)$. 
\end{remark}

In the following proposition, we shows that when $p=2, n \geq 6$ and $m=\frac{n}{2}$, $h(0)=0$.

\begin{proposition}\label{Proposition 1}
Let $p=2$ and $n \geq 6$ be an even integer. Suppose $\{A_{i}, i \in V_{n/2}\}$ is a bent partition of $V_{n}$. Then the function $h: V_{n} \rightarrow \mathbb{F}_{2}$ in Theorem \ref{Theorem 1} (ii) satisfies $h(0)=0$, that is, $|A_{G(0)}|=2^{\frac{n}{2}+1}-1, |A_{i}|=2^{\frac{n}{2}}-1, i \neq G(0)$, where $G: V_{n} \rightarrow V_{n/2}$ is given in Theorem \ref{Theorem 1} (ii).
\end{proposition}

\begin{proof}
To the contrary, assume that $h(0)=1$. By Remark \ref{Remark 1}, $|A_{G(0)}|=1, |A_{i}|=2^{\frac{n}{2}}+1, i \neq G(0)$, and $D_{G, G(0)}=\{0\}, |D_{G, i}|=2^{\frac{n}{2}}+1, i \neq G(0)$. Denote $A_{G(0)}=\{\alpha\}$. Since $D_{G, G(0)}=\{0\}$, by Eq. \eqref{R1}, for any $a \in V_{n}^{*}$, we have $(-1)^{\langle \alpha, a\rangle_{n}}=\chi_{a}(A_{G(0)})=(-1)^{h(a)+1}$. Then $h(x)=\langle \alpha, x \rangle_{n}+1, x \in V_{n}$. For any $i \neq G(0)$, by $h(x)=\langle \alpha, x \rangle_{n}+1, |A_{i}|=|D_{G, i}^{*}|=2^{\frac{n}{2}}+1$ and Eq. \eqref{R1}, we have $\chi_{a}(A_{i}+\alpha)=(-1)^{\langle\alpha, a\rangle_{n}}\chi_{a}(A_{i})=1-2^\frac{n}{2}\delta_{G(a)}(i)$ for any $a \in V_{n}^{*}$, and 
\begin{equation*}
\sum_{a \in V_{n}}\big(\chi_{a}(A_{i}+\alpha)\big)^{3}=(2^\frac{n}{2}+1)^{3}+(2^{\frac{n}{2}}+1)(1-2^{\frac{n}{2}})^{3}+2^{n}-1-(2^{\frac{n}{2}}+1)=2^{n}(2^{\frac{n}{2}}+1)(4-2^{\frac{n}{2}}) <0,
\end{equation*}
which contradicts the fact that
\begin{equation*}
	\begin{split}
	\sum_{a \in V_{n}}\big(\chi_{a}(A_{i}+\alpha)\big)^{3}&=\sum_{a \in V_{n}}\sum_{x \in A_{i}+\alpha}(-1)^{\langle a, x \rangle_{n}}\sum_{y \in A_{i}+\alpha}(-1)^{\langle a, y \rangle_{n}}\sum_{z \in A_{i}+\alpha}(-1)^{\langle a, z \rangle_{n}}\\
	&=\sum_{x, y, z \in A_{i}+\alpha}\sum_{a \in V_{n}}(-1)^{\langle a, x+y+z\rangle_{n}}\\
	&=2^{n}|\{(x, y): x, y, x+y \in A_{i}+\alpha\}|\geq 0.
	\end{split}
\end{equation*}
Thus, $h(0)=0$.
\end{proof}

\begin{remark}
When $p=2, n=4, m=2$, by \cite[Proposition 4]{AFKMWW2026A}, the preimage set partition $\{D_{F, i}, i \in V_{2}\}$ of any vectorial bent function $F: V_{4} \rightarrow V_{2}$ is a bent partition. There exist vectorial bent functions $F$ such that the function $h: V_{n} \rightarrow \mathbb{F}_{2}$ in Theorem \ref{Theorem 1} (ii) satisfies $h(0)=1$, that is, $|D_{F, i_{0}}|=1$ for some $i_{0} \in V_{2}$. For example, by \cite[Theorem 6.32]{LN1997Fi}, the vectorial bent function $F: \mathbb{F}_{2^2}^{2} \rightarrow \mathbb{F}_{2^2}$ defined by
 $F(x_{1}, x_{2})=x_{1}x_{2}+x_{1}^{2}+ax_{2}^{2}, x_{1}, x_{2} \in \mathbb{F}_{2^2}$ satisfies $|D_{F, 0}|=1$, where $a \in \mathbb{F}_{2^2}$ with $\mathrm{Tr}_{1}^{2}(a)=1$.
\end{remark}

The following characterization of bent partitions in class $\mathcal{WRBP}$ via Buston Hadamard matrices directly follows from Theorem \ref{Theorem 1} and \cite[Proposition 2]{WWF2026Fu}.

\begin{corollary}\label{Corollary 1}
Let $\Gamma=\{A_{i}, i \in V_{m}\}$ be a partition of $V_{n}$, where $n$ is even, $m \leq \frac{n}{2}$, and $m \geq 2$ if $p=2$. Let $F: V_{n} \rightarrow V_{m}$ be given by $F(x)=i, x \in A_{i}, i \in V_{m}$. Define
\begin{equation*}
H_{c}=\left[\zeta_{p}^{\langle c, F(x-y)\rangle_{m}}\right]_{x, y \in V_{n}}, c \in V_{m}^{*}.
\end{equation*}
Then the following statements are equivalent.

(i) $\Gamma$ is a bent partition in class $\mathcal{WRBP}$.

(ii) $H_{c}, c \in V_{m}^{*}$, are Buston Hadamard matrices of weakly regular type with $\vartheta_{H_{c}}=\epsilon$ and $H_{d}H_{e}^{*}H_{d-e}^{*}, d \neq e \in V_{m}^{*}$, are all the same, where $\epsilon \in \{\pm 1 \}$ with $\epsilon=1$ if $p \neq 3$.
\end{corollary}

In the following, we characterize bent partitions in class $\mathcal{WRBP}$ in terms of linear codes.

\begin{theorem} \label{Theorem 2}
Let $\Gamma=\{A_{i}, i \in V_{m}\}$ be a partition of $V_{n}$, where $n$ is even, $m \leq \frac{n}{2}$, and $m \geq 2$ if $p=2$. Let $F: V_{n} \rightarrow V_{m}$ be given by $F(x)=i, x \in A_{i}, i \in V_{m}$.
Then the following statements are equivalent.

(i) $\Gamma$ is a bent partition in class $\mathcal{WRBP}$.

(ii)  Define linear code $\mathcal{C}_{\Gamma}$ as follows:
\begin{equation*}
	\mathcal{C}_{\Gamma}=\{c_{\alpha, \beta, \gamma}=(\langle \alpha, F(x)\rangle_{m}+\langle \beta, x \rangle_{n}+\gamma)_{x \in V_{n}}: \alpha \in V_{m}, \beta \in V_{n}, \gamma \in \mathbb{F}_{p}\}.
\end{equation*}
Then $\mathcal{C}_{\Gamma}$ is a four weight $[p^{n}, n+m+1]_{p}$ linear code whose weight enumerator is
\begin{equation*}
1+(p-1)z^{p^n}+p(p^{n}-1)z^{p^{n-1}(p-1)}+p^{n}(p^{m}-1)z^{p^{n}-p^{n-1}-\epsilon p^{\frac{n}{2}-1}(p-1)}+p^{n}(p^{m}-1)(p-1)z^{p^{n}-p^{n-1}+\epsilon p^{\frac{n}{2}-1}},
\end{equation*}
and for any $\beta \in V_{n}$, there exists $\gamma_{\beta} \in \mathbb{F}_{p}$ such that $(\alpha+\alpha',  \gamma+\gamma'-\gamma_{\beta}) \in S_{\beta}$ whenever $(\alpha, \gamma), (\alpha', \gamma') \in S_{\beta}$ with $\alpha \neq -\alpha' \in V_{m}^{*}$, where $\epsilon \in \{\pm 1\}$ with $\epsilon=1$ if $p \neq 3$, and
\begin{equation*}
S_{\beta}=\{(\alpha, \gamma) \in V_{m}^{*} \times \mathbb{F}_{p}: wt(c_{\alpha, \beta, \gamma})=p^{n}-p^{n-1}-\epsilon p^{\frac{n}{2}-1}(p-1)\}.
\end{equation*}

(iii) For any $i \in V_{m}$, define linear code $\mathcal{C}_{A_{i}}$ as follows:
\begin{equation*}
	\mathcal{C}_{A_{i}}=\{c_{a, b}^{(i)}=(\langle a, x \rangle_{n}+b)_{x \in A_{i}}: a \in V_{n}, b \in \mathbb{F}_{p}\}.
\end{equation*}
Then there exists a function $\tau: V_{n} \rightarrow \mathbb{F}_{p}$ with $\tau(0)=0$ if $p \neq 2$ such that for any fixed $i \in V_{m}$ and $a \in V_{n}^{*}$, 
\begin{equation}\label{th2new1}
	wt(c_{a, b}^{(i)})=
	\begin{cases}
	\frac{(p-1)|A_{i}|-\epsilon p^{\frac{n}{2}-m}(p^m-1)(p-1)}{p}, & \text{ if } b=\tau(a),\\
	\frac{(p-1)|A_{i}|+\epsilon p^{\frac{n}{2}-m}(p^m-1)}{p}, & \text{ if } b \in \mathbb{F}_{p} \backslash \{\tau(a)\},
	\end{cases}
\end{equation}
or
\begin{equation}\label{th2new2}
wt(c_{a, b}^{(i)})=
   	\begin{cases}
   	\frac{(p-1)|A_{i}|+\epsilon p^{\frac{n}{2}-m}(p-1)}{p}, & \text{ if } b=\tau(a),\\
   	\frac{(p-1)|A_{i}|-\epsilon p^{\frac{n}{2}-m}}{p}, & \text{ if } b \in \mathbb{F}_{p} \backslash \{\tau(a)\},
   	\end{cases}
\end{equation}
where $\epsilon \in \{\pm 1\}$ with $\epsilon=1$ if $p \neq 3$, and
$|A_{i}|=p^{n-m}+\epsilon  p^{\frac{n}{2}-m}(p^m-1)\zeta_{p}^{\tau(0)}$ or $p^{n-m}-\epsilon p^{\frac{n}{2}-m}\zeta_{p}^{\tau(0)}$.
\end{theorem}

\begin{proof}
For $w \in \mathbb{F}_{p}^{*}$, if $p=2$, let $\sigma_{w}$ denote the identity map of $\mathbb{Q}$, and if $p$ is odd, let $\sigma_{w}$ denote the automorphism of cyclotomic field $\mathbb{Q}(\zeta_{p})$ defined by $\sigma_{w}(\zeta_{p})=\zeta_{p}^{w}$.

$(i) \Rightarrow (ii)$: By Theorem \ref{Theorem 1}, $F_{c}, c \in V_{m}^{*}$, are all weakly regular bent with $\epsilon_{F_{c}}=\epsilon$, where $\epsilon\in \{\pm 1\}$ with $\epsilon=1$ if $p \neq 3$, and there exist functions $G: V_{n} \rightarrow V_{m}$ and $h: V_{n} \rightarrow \mathbb{F}_{p}$ such that $(F_{c})^{*}=G_{c}+h$ for any $c \in V_{m}^{*}$. Obviously, when $\alpha=0, \beta=0, \gamma \in \mathbb{F}_{p}^{*}$, $wt(c_{\alpha, \beta, \gamma})=p^{n}$; when $\alpha=0, \beta \in V_{n}^{*}, \gamma \in \mathbb{F}_{p}$, $wt(c_{\alpha, \beta, \gamma})=p^{n-1}(p-1)$. For $(\alpha, \beta, \gamma) \in V_{m}^{*} \times V_{n} \times \mathbb{F}_{p}$, we have
$wt(c_{\alpha, \beta, \gamma})=p^n-N_{\alpha, \beta,
\gamma}$, where
\begin{equation*}
	\begin{split}
		N_{\alpha, \beta, \gamma} &\triangleq |\{x \in V_{n}: \langle \alpha, F(x) \rangle_{m}+\langle \beta, x \rangle_{n}+\gamma=0\}|\\
		&=\frac{1}{p}\sum_{x \in V_{n}}\sum_{w \in \mathbb{F}_{p}}\zeta_{p}^{w(F_{\alpha}(x)+\langle \beta, x \rangle_{n}+\gamma)}\\
			\end{split}
	\end{equation*}
\begin{equation*}
\begin{split}
		&=p^{n-1}+\frac{1}{p}\sum_{w \in \mathbb{F}_{p}^{*}}\zeta_{p}^{w \gamma} \sum_{x \in V_{n}}\zeta_{p}^{w(F_{\alpha}(x)+\langle \beta, x \rangle_{n})}\\
		&=p^{n-1}+\frac{1}{p}\sum_{w \in \mathbb{F}_{p}^{*}}\zeta_{p}^{w \gamma}\sigma_{w}(W_{F_{\alpha}}(-\beta))\\
		&=p^{n-1}+\frac{1}{p}\sum_{w \in \mathbb{F}_{p}^{*}}\zeta_{p}^{w \gamma}\sigma_{w}(\epsilon p^{\frac{n}{2}}\zeta_{p}^{(F_{\alpha})^{*}(-\beta)})\\
		&=p^{n-1}+\epsilon p^{\frac{n}{2}-1}\sum_{w \in \mathbb{F}_{p}^{*}}\zeta_{p}^{w(G_{\alpha}(-\beta)+h(-\beta)+\gamma)}\\
		&=\begin{cases}
			p^{n-1}+\epsilon p^{\frac{n}{2}-1}(p-1), & \text{ if } G_{\alpha}(-\beta)+h(-\beta)+\gamma=0,\\
			p^{n-1}-\epsilon p^{\frac{n}{2}-1}, & \text{ if } G_{\alpha}(-\beta)+h(-\beta)+\gamma \neq 0.
		\end{cases}
	\end{split}
\end{equation*}
Hence, when $\alpha \in V_{m}^{*}, \beta \in V_{n}, \gamma=-G_{\alpha}(-\beta)-h(-\beta)$, we have $wt(c_{\alpha, \beta, \gamma})=p^{n}-p^{n-1}-\epsilon p^{\frac{n}{2}-1}(p-1)$; when $\alpha \in V_{m}^{*}, \beta \in V_{n}, \gamma \neq -G_{\alpha}(-\beta)-h(-\beta)$, we have $wt(c_{\alpha, \beta, \gamma})=p^{n}-p^{n-1}+\epsilon p^{\frac{n}{2}-1}$. For any $\beta \in V_{n}$, let $\gamma_{\beta}=-h(-\beta)$. If $(\alpha, \gamma), (\alpha', \gamma') \in S_{\beta}$ with $\alpha \neq -\alpha' \in V_{m}^{*}$, then $\gamma+\gamma'-\gamma_{\beta}=-G_{\alpha}(-\beta)+\gamma_{\beta}-G_{\alpha'}(-\beta)+\gamma_{\beta}-\gamma_{\beta}=-G_{\alpha+\alpha'}(-\beta)+\gamma_{\beta}$, and thus $(\alpha+\alpha', \gamma+\gamma'-\gamma_{\beta}) \in S_{\beta}$.

$(ii) \Rightarrow (i)$: We claim that for any $(\alpha, \beta, \gamma) \in V_{m}^{*} \times V_{n} \times \mathbb{F}_{p}$,
\begin{equation*}
	wt(c_{\alpha, \beta, \gamma})=p^{n}-p^{n-1}-\epsilon p^{\frac{n}{2}-1}(p\delta_{\tau_{\alpha, \beta}}(\gamma)-1),
\end{equation*}
where $\tau_{\alpha, \beta} \in \mathbb{F}_{p}$ satisfies that for any $\beta \in \mathbb{F}_{p}$, $\tau_{\alpha, \beta}+\tau_{\alpha', \beta}-\tau_{\alpha+\alpha', \beta}=\gamma_{\beta}$ for all $\alpha \neq -\alpha' \in V_{m}^{*}$.
Obviously, when $\alpha=0, \beta=0, \gamma \in \mathbb{F}_{p}^{*}$, $wt(c_{\alpha, \beta, \gamma})=p^{n}$; when $\alpha=0, \beta \in V_{n}^{*}, \gamma \in \mathbb{F}_{p}$, $wt(c_{\alpha, \beta, \gamma})=p^{n-1}(p-1)$. Then by the weight enumerator of $\mathcal{C}_{\Gamma}$, for any $(\alpha, \beta, \gamma) \in V_{m}^{*} \times V_{n} \times \mathbb{F}_{p}$, $wt(c_{\alpha, \beta, \gamma})=p^{n}-p^{n-1}-\epsilon p^{\frac{n}{2}-1}(p-1)$ or $p^{n}-p^{n-1}+\epsilon p^{\frac{n}{2}-1}$. We first show that for any $(\alpha, \beta) \in V_{m}^{*} \times V_{n}$, there is unique $\tau_{\alpha, \beta} \in \mathbb{F}_{p}$ such that $wt(c_{\alpha, \beta, \gamma})=p^{n}-p^{n-1}-\epsilon p^{\frac{n}{2}-1}(p\delta_{\tau_{\alpha, \beta}}(\gamma)-1)$. To the contrary, assume that there is $(\alpha_{0}, \beta_{0}) \in V_{m}^{*} \times V_{n}$ such that $wt(c_{\alpha_{0}, \beta_{0}, \gamma})=p^{n}-p^{n-1}+\epsilon p^{\frac{n}{2}-1}$ for all $\gamma \in \mathbb{F}_{p}$, or  $wt(c_{\alpha_{0}, \beta_{0}, \gamma})=wt(c_{\alpha_{0}, \beta_{0}, \gamma'})=p^{n}-p^{n-1}-\epsilon p^{\frac{n}{2}-1}(p-1)$ for some $\gamma \neq \gamma' \in \mathbb{F}_{p}$. If $wt(c_{\alpha_{0}, \beta_{0}, \gamma})=p^{n}-p^{n-1}+\epsilon p^{\frac{n}{2}-1}$ for all $\gamma \in \mathbb{F}_{p}$, then for any $\gamma \in \mathbb{F}_{p}$,
\begin{equation*}
N_{\alpha_{0}, \beta_{0}, \gamma}=|\{x \in V_{n}: \langle \alpha_{0}, F(x) \rangle_{m}+\langle \beta_{0}, x \rangle_{n}+\gamma=0\}| =p^{n-1}-\epsilon p^{\frac{n}{2}-1}
,
\end{equation*}
which contradicts the fact that
\begin{equation*}
\sum_{\gamma \in \mathbb{F}_{p}} N_{\alpha_{0}, \beta_{0}, \gamma}=p^n.
\end{equation*}
 Thus, for any $(\alpha, \beta) \in V_{m}^{*} \times V_{n}$, there is $\gamma \in \mathbb{F}_{p}$ such that $wt(c_{\alpha, \beta, \gamma})=p^{n}-p^{n-1}-\epsilon p^{\frac{n}{2}-1}(p-1)$. If $wt(c_{\alpha_{0}, \beta_{0}, \gamma})=wt(c_{\alpha_{0}, \beta_{0}, \gamma'})=p^{n}-p^{n-1}-\epsilon p^{\frac{n}{2}-1}(p-1)$ for some $\gamma \neq \gamma' \in \mathbb{F}_{p}$, then
 \begin{equation*}
 |\{(\alpha, \beta, \gamma) \in V_{m} \times V_{n} \times \mathbb{F}_{p}: wt(c_{\alpha, \beta, \gamma})=p^{n}-p^{n-1}-\epsilon p^{\frac{n}{2}-1}(p-1)\}|\geq (p^{m}-1)p^{n}+1,
 \end{equation*}
 which contradicts the weight enumerator of $\mathcal{C}_{\Gamma}$. Therefore, for each $(\alpha, \beta) \in V_{m}^{*} \times V_{n}$, there is unique $\tau_{\alpha, \beta} \in \mathbb{F}_{p}$ such that $wt(c_{\alpha, \beta, \gamma})=p^{n}-p^{n-1}-\epsilon p^{\frac{n}{2}-1}(p\delta_{\tau_{\alpha, \beta}}(\gamma)-1), \gamma \in \mathbb{F}_{p}$. Note that  $(\alpha, \gamma) \in S_{\beta}$ if and only if $\gamma=\tau_{\alpha, \beta}$. Since for any $\beta \in V_{n}$, there exists $\gamma_{\beta} \in \mathbb{F}_{p}$ such that $(\alpha+\alpha',  \gamma+\gamma'-\gamma_{\beta}) \in S_{\beta}$ whenever $(\alpha, \gamma), (\alpha', \gamma') \in S_{\beta}$ with $\alpha \neq -\alpha' \in V_{m}^{*}$, then $\tau_{\alpha, \beta}+\tau_{\alpha', \beta}-\tau_{\alpha+\alpha', \beta}=\gamma_{\beta}$ for all $\alpha \neq -\alpha' \in V_{m}^{*}$. Hence, the claim holds.

For any $(\alpha, \beta, \gamma) \in V_{m}^{*} \times V_{n} \times \mathbb{F}_{p}$,
\begin{equation*}
	\begin{split}
		N_{\alpha, \beta, \gamma}
		&=|\{x \in V_{n}: \langle \alpha, F(x) \rangle_{m}+\langle \beta, x \rangle_{n}+\gamma=0\}|\\
		&=\frac{1}{p}\sum_{x \in V_{n}}\sum_{w \in \mathbb{F}_{p}}\zeta_{p}^{w(\langle \alpha, F(x)\rangle_{m}+\langle \beta, x \rangle_{n}+\gamma)}\\
		&=p^{n-1}+\frac{1}{p}\sum_{w \in \mathbb{F}_{p}^{*}}\zeta_{p}^{w \gamma}W_{F_{w \alpha}(-w\beta)}\\
		&=p^{n-1}+\frac{1}{p}\sum_{w \in \mathbb{F}_{p}}\zeta_{p}^{w \gamma}T^{(\alpha, \beta)}(w)\\
		&=p^{n-1}+\frac{1}{p}
		\widehat{T^{(\alpha, \beta)}}(\gamma),
	\end{split}
\end{equation*}
where
\begin{equation*}
T^{(\alpha, \beta)}(x)=\begin{cases}
	0, & \text{ if } x=0,\\
	W_{F_{x\alpha}}(-x\beta), & \text{ if } x \in \mathbb{F}_{p}^{*},
\end{cases}
\end{equation*}
and $\widehat{T^{(\alpha, \beta)}}$ is the Fourier transform of $T^{(\alpha, \beta)}$. Thus, $wt(c_{\alpha, \beta, \gamma})=p^{n}-p^{n-1}-\frac{1}{p}\widehat{T^{(\alpha, \beta)}}(\gamma)$. By the inverse Fourier transform, for $x \in \mathbb{F}_{p}^{*}$,
\begin{equation*}
		T^{(\alpha, \beta)}(x)=\frac{1}{p}\sum_{\gamma \in \mathbb{F}_{p}}\widehat{T^{(\alpha, \beta)}}(\gamma)\zeta_{p}^{-\gamma x}=\sum_{\gamma \in \mathbb{F}_{p}}(p^n-p^{n-1}-wt(c_{\alpha, \beta, \gamma}))\zeta_{p}^{-\gamma x}=-\sum_{\gamma \in \mathbb{F}_{p}}wt(c_{\alpha, \beta, \gamma})\zeta_{p}^{-\gamma x}.
\end{equation*}
Hence, for any $(\alpha, \beta) \in V_{m}^{*} \times V_{n}$, we have
\begin{equation*}
	\begin{split}
	W_{F_{\alpha}}(-\beta)=T^{(\alpha, \beta)}(1)
	&=-\sum_{\gamma \in \mathbb{F}_{p}}wt(c_{\alpha, \beta, \gamma})\zeta_{p}^{-\gamma}\\
	&=-\sum_{\gamma \in \mathbb{F}_{p}}\Big(p^{n}-p^{n-1}-\epsilon p^{\frac{n}{2}-1}(p\delta_{\tau_{\alpha, \beta}}(\gamma)-1)\Big)\zeta_{p}^{-\gamma}\\
	&=\epsilon p^{\frac{n}{2}}\zeta_{p}^{-\tau_{\alpha, \beta}},
	\end{split}
\end{equation*}
which implies that $F$ is a vectorial bent function for which any component function $F_{\alpha}$ ($\alpha \in V_{m}^{*}$) is weakly regular with $\epsilon_{F_{\alpha}}=\epsilon$ and  $(F_{\alpha})^{*}(x)=-\tau_{\alpha, -x}, x \in V_{n}$. For any $x \in V_{n}$ and $\alpha \neq -\alpha' \in V_{m}^{*}$,  since $\tau_{\alpha, -x}+\tau_{\alpha', -x}-\tau_{\alpha+\alpha', -x}=\gamma_{-x}$, then $(F_{\alpha})^{*}(x)+(F_{\alpha'})^{*}(x)-(F_{\alpha+\alpha'})^{*}(x)=-\gamma_{-x}$. By Theorem \ref{Theorem 1}, $\Gamma$ is a bent partition in class $\mathcal{WRBP}$.

$(i) \Rightarrow (iii)$:
By Theorem \ref{Theorem 1} and its proof and Remark \ref{Remark 1}, there are functions $G: V_{n} \rightarrow V_{m}$ and $\tau: V_{n} \rightarrow \mathbb{F}_{p}$ with $\tau(0)=0$ if $p \neq 2$ such that
\begin{equation*}
	\chi_{-a}(A_{i})=p^{n-m}\delta_{0}(a)+\epsilon p^{\frac{n}{2}-m}\zeta_{p}^{\tau(a)}(p^m \delta_{G(a)}(i)-1),
\end{equation*}
where $\epsilon \in \{\pm 1\}$ with $\epsilon=1$ if $p \neq 3$. For $a \in V_{n}^{*}, b \in \mathbb{F}_{p}$ and $i \in V_{m}$,
\begin{equation*}
	\begin{split}
		N_{a, b}^{(i)}&\triangleq|\{x \in A_{i}: \langle a, x \rangle_{n}+b=0 \}|\\
		&=\frac{1}{p}\sum_{x \in A_{i}}\sum_{w \in \mathbb{F}_{p}}\zeta_{p}^{-w(\langle a, x\rangle_{n}+b)}\\
		&=\frac{1}{p}|A_{i}|+\frac{1}{p}\sum_{w \in \mathbb{F}_{p}^{*}}\zeta_{p}^{-wb}\sum_{x \in A_{i}}\zeta_{p}^{w\langle -a, x\rangle_{n}}\\
		&=\frac{1}{p}|A_{i}|+\frac{1}{p}\sum_{w \in \mathbb{F}_{p}^{*}}\zeta_{p}^{-wb}\sigma_{w}(\chi_{-a}(A_{i}))\\
		&=\frac{1}{p}|A_{i}|+\frac{1}{p}\sum_{w \in \mathbb{F}_{p}^{*}}\zeta_{p}^{-wb}\sigma_{w}\Big(\epsilon p^{\frac{n}{2}-m}\zeta_{p}^{\tau(a)}(p^m \delta_{G(a)}(i)-1)\Big)\\
		&=\frac{1}{p}|A_{i}|+\frac{1}{p}\sum_{w \in \mathbb{F}_{p}^{*}}\zeta_{p}^{-wb}\Big(\epsilon p^{\frac{n}{2}-m}\zeta_{p}^{w\tau(a)}(p^m \delta_{G(a)}(i)-1)\Big)\\
		&=\frac{1}{p}|A_{i}|+\epsilon p^{\frac{n}{2}-m-1}(p^m \delta_{G(a)}(i)-1)\sum_{w \in \mathbb{F}_{p}^{*}}\zeta_{p}^{w(-b+\tau(a))}\\
	\end{split}
	\end{equation*}
	\begin{equation}\label{addth21}
	\begin{split}
		&=\begin{cases}
			\frac{|A_{i}|+\epsilon p^{\frac{n}{2}-m}(p^m-1)(p-1)}{p}, & \text{ if } G(a)=i, b=\tau(a),\\
			\frac{|A_{i}|-\epsilon p^{\frac{n}{2}-m}(p^m-1)}{p}, & \text{ if } G(a)=i, b \neq \tau(a),\\
			\frac{|A_{i}|-\epsilon p^{\frac{n}{2}-m}(p-1)}{p}, & \text{ if } G(a) \neq i, b=\tau(a),\\
			\frac{|A_{i}|+\epsilon p^{\frac{n}{2}-m}}{p}, & \text{ if } G(a) \neq i, b \neq \tau(a).
		\end{cases}
	\end{split}
\end{equation}
Thus, for any fixed $i \in V_{m}$ and $a \in V_{n}^{*}$, Eq. \eqref{th2new1} or Eq. \eqref{th2new2} holds.

$(iii) \Rightarrow (i)$: For any $a \in V_{n}^{*}$ and $i \in V_{m}$, define
\begin{equation*}
	T^{(a, i)}(x)=\begin{cases}
		0, & \text{ if } x=0,\\
		\sigma_{x}(\chi_{-a}(A_{i})), & \text{ if } x \in \mathbb{F}_{p}^{*}.
	\end{cases}
\end{equation*}
For $a \in V_{n}^{*}, b \in \mathbb{F}_{p}, i \in V_{m}$,
\begin{equation*}
N_{a, b}^{(i)}=|\{x \in A_{i}: \langle a, x \rangle_{n}+b=0 \}|=\frac{1}{p}|A_{i}|+\frac{1}{p}\sum_{w \in \mathbb{F}_{p}}\zeta_{p}^{-wb}T^{(a, i)}(w)=\frac{1}{p}|A_{i}|+\frac{1}{p}\widehat{T^{(a, i)}}(-b).
\end{equation*}
Then for $a \in V_{n}^{*}, i \in V_{m}$ and $x \in \mathbb{F}_{p}^{*}$, by the inverse Fourier transform, we have
\begin{equation*}
	T^{(a, i)}(x)=\frac{1}{p}\sum_{b \in \mathbb{F}_{p}}\widehat{T^{(a, i)}}(-b)\zeta_{p}^{bx}=\sum_{b \in \mathbb{F}_{p}}\Big(\frac{p-1}{p}|A_{i}|-wt(c_{a, b}^{(i)})\Big)\zeta_{p}^{bx}=-\sum_{b \in \mathbb{F}_{p}}wt(c_{a, b}^{(i)})\zeta_{p}^{bx}.
\end{equation*}
Let $x=1$, for any fixed $a \in V_{n}^{*}$ and $i \in V_{m}$, since Eq. \eqref{th2new1} or Eq. \eqref{th2new2} holds, then we have
\begin{equation*}
	\begin{split}
		\chi_{-a}(A_{i})&=T^{(a, i)}(1)=-\sum_{b \in \mathbb{F}_{p}}wt(c_{a, b}^{(i)})\zeta_{p}^{b}\\
		& \in \Big\{-\sum_{b \in \mathbb{F}_{p}}\frac{(p-1)|A_{i}|-\epsilon p^{\frac{n}{2}-m}(p^m-1)(p \delta_{\tau(a)}(b)-1)}{p}\zeta_{p}^{b}, \\
		& \ \ \ \ \ -\sum_{b \in \mathbb{F}_{p}}\frac{(p-1)|A_{i}|+\epsilon p^{\frac{n}{2}-m}(p \delta_{\tau(a)}(b)-1)}{p}\zeta_{p}^{b}\Big\}\\
		&=\big\{\epsilon p^{\frac{n}{2}-m}(p^m-1)\zeta_{p}^{\tau(a)}, -\epsilon p^{\frac{n}{2}-m}\zeta_{p}^{\tau(a)}\big\}.
	\end{split}
\end{equation*}
Then for any $a \in V_{n}$ and $i \in V_{m}$,
\begin{equation*}
	\chi_{-a}(A_{i})=p^{n-m}\delta_{0}(a)+\epsilon p^{\frac{n}{2}-m}(p^m-1)\zeta_{p}^{\tau(a)} \text{ or } p^{n-m}\delta_{0}(a)-\epsilon p^{\frac{n}{2}-m}\zeta_{p}^{\tau(a)}.
\end{equation*}
For any $a \in V_{n}$, let $n_{a}=|\{i \in V_{m}: \chi_{-a}(A_{i})=p^{n-m}\delta_{0}(a)+\epsilon p^{\frac{n}{2}-m}(p^m-1)\zeta_{p}^{\tau(a)}\}|$. Then
\begin{equation*}
	\begin{split}
	p^{n}\delta_{0}(a)&=\chi_{-a}(V_{n})\\
	&=\sum_{i \in V_{m}}\chi_{-a}(A_{i})\\
	&=n_{a}\big(p^{n-m}\delta_{0}(a)+\epsilon p^{\frac{n}{2}-m}\zeta_{p}^{\tau(a)}(p^m-1)\big)+(p^m-n_{a})\big(p^{n-m}\delta_{0}(a)-\epsilon p^{\frac{n}{2}-m}\zeta_{p}^{\tau(a)}\big),
	\end{split}
\end{equation*}
which implies that $n_{a}=1$. Therefore, there exists $G: V_{n} \rightarrow V_{m}$ such that for any $a \in V_{n}$ and $i \in V_{m}$,
\begin{equation}\label{th24}
	\chi_{-a}(A_{i})=p^{n-m}\delta_{0}(a)+\epsilon p^{\frac{n}{2}-m}(p^m \delta_{G(a)}(i)-1)\zeta_{p}^{\tau(a)}.
\end{equation}
Denote $F(x)=i, x \in A_{i}, i \in V_{m}$. For any balanced map $B$ from $V_{m}$ to $\mathbb{F}_{p}$ and $a \in V_{n}$, by Eq. \eqref{th24},
\begin{equation*}
W_{B(F)}(a)=\sum_{x \in V_{n}}\zeta_{p}^{B(F(x))-\langle a, x \rangle_{n}}=\sum_{i \in V_{m}}\zeta_{p}^{B(i)}\chi_{-a}(A_{i})=\epsilon p^{\frac{n}{2}}\zeta_{p}^{B(G(a))+\tau(a)}.
\end{equation*}
Then $B(F(x))$ is a weakly regular bent function with $\epsilon_{B(F)}=\epsilon$, and thus $\Gamma$ is a bent partition in class $\mathcal{WRBP}$.
\end{proof}

\begin{remark}
Keep the same notations as in Theorem \ref{Theorem 2}. Suppose $\Gamma=\{A_{i}, i \in V_{m}\}$ is a bent partition of $V_{n}$ in class $\mathcal{WRBP}$, where $m \geq 2$ if $p=2$. Note that $\tau(0)=0$ if $p$ is odd, or $p=2$ and $(m, n)$ with $3\leq m=\frac{n}{2}$ by Theorem \ref{Theorem 2} and Proposition \ref{Proposition 1}. By Eq. \eqref{addth21} and the value distribution of $G$ given in Eq. \eqref{co5} for odd $p$ and Remark \ref{Remark 1} for $p=2$, the parameters and weight enumerators of $\mathcal{C}_{A_{i}}, i \in V_{m},$ can be readily obtained as follows.
\begin{itemize}
\item Denote $l_{1}=p^{n-m}-\epsilon p^{\frac{n}{2}-m}\zeta_{p}^{\tau(0)}$. Except for the case where $m=\frac{n}{2}, \epsilon=1, \tau(0)=0$, $\mathcal{C}_{A_{i}} (i \neq G(0))$ is a $[l_{1}, n+1]_{p}$ linear code; furthermore, when $m \geq 2$, $\mathcal{C}_{A_{i}} (i \neq G(0))$ is $5$-weight and its weight enumerator is $1+l_{1}z^{\frac{(p-1)l_{1}-\epsilon p^{\frac{n}{2}-m}(p^m-1)(p-1)}{p}}+(p-1)l_{1}z^{\frac{(p-1)l_{1}+\epsilon p^{\frac{n}{2}-m}(p^m-1)}{p}}+(p^n-1-l_{1})z^{\frac{(p-1)l_{1}+\epsilon p^{\frac{n}{2}-m}(p-1)}{p}}+(p-1)(p^n-1-l_{1})z^{\frac{(p-1)l_{1}-\epsilon p^{\frac{n}{2}-m}}{p}}+(p-1)z^{l_{1}}$, and when $m=1$ (in this case, $p$ is odd), $\mathcal{C}_{A_{i}} (i \neq G(0))$ is $4$-weight and its weight enumerator is $1+l_{1}z^{\frac{(p-1)l_{1}-\epsilon p^{\frac{n}{2}-1}(p-1)^{2}}{p}}+(p^n-1+(p-2)l_{1})z^{\frac{(p-1)l_{1}+\epsilon p^{\frac{n}{2}-1}(p-1)}{p}}+(p-1)(p^n-1-l_{1})z^{\frac{(p-1)l_{1}-\epsilon p^{\frac{n}{2}-1}}{p}}+(p-1)z^{l_{1}}$. In the case where $m=\frac{n}{2}, \epsilon=1, \tau(0)=0$, $\mathcal{C}_{A_{i}} (i \neq G(0))$ is a $[p^{\frac{n}{2}}-1, \frac{n}{2}+1]_{p}$ linear code; furthermore, when $m \geq 2$, $\mathcal{C}_{A_{i}} (i \neq G(0))$ is $3$-weight and its weight enumerator is $1+(p-1)z^{p^{\frac{n}{2}}-1}+(p^\frac{n}{2}-1)z^{(p-1)p^{\frac{n}{2}-1}}+(p-1)(p^{\frac{n}{2}}-1)z^{p^{\frac{n}{2}}-p^{\frac{n}{2}-1}-1}$, and when $m=1$ (in this case, $n=2$ and $p$ is odd), $\mathcal{C}_{A_{i}} (i \neq G(0))$ is $2$-weight and its weight enumerator is $1+2(p-1)z^{p-1}+(p-1)^{2}z^{p-2}$.
\item Denote $l_{2}=p^{n-m}+\epsilon p^{\frac{n}{2}-m}(p^m-1)\zeta_{p}^{\tau(0)}$. Except for the case where $p=2, m=2, n=4, \tau(0)=1$, $\mathcal{C}_{A_{G(0)}}$ is a $[l_{2}, n+1]_{p}$ linear code; furthermore, when $m \geq 2$, $\mathcal{C}_{A_{G(0)}}$ is $5$-weight and its weight enumerator is $1+(l_{2}-1)z^{\frac{(p-1)l_{2}-\epsilon p^{\frac{n}{2}-m}(p^m-1)(p-1)}{p}}+(p-1)(l_{2}-1)z^{\frac{(p-1)l_{2}+\epsilon p^{\frac{n}{2}-m}(p^m-1)}{p}}+(p^n-l_{2})z^{\frac{(p-1)l_{2}+\epsilon p^{\frac{n}{2}-m}(p-1)}{p}}+(p-1)(p^n-l_{2})z^{\frac{(p-1)l_{2}-\epsilon p^{\frac{n}{2}-m}}{p}}+(p-1)z^{l_{2}}$, and when $m=1$ (in this case, $p$ is odd), $\mathcal{C}_{A_{G(0)}}$ is $4$-weight and its weight enumerator is $1+(l_{2}-1)z^{\frac{(p-1)l_{2}-\epsilon p^{\frac{n}{2}-1}(p-1)^{2}}{p}}+(p^n-p+1+(p-2)l_{2})z^{\frac{(p-1)l_{2}+\epsilon p^{\frac{n}{2}-1}(p-1)}{p}}+(p-1)(p^n-l_{2})z^{\frac{(p-1)l_{2}-\epsilon p^{\frac{n}{2}-1}}{p}}+(p-1)z^{l_{2}}$. In the case where $p=2, m=2, n=4, \tau(0)=1$, $\mathcal{C}_{A_{G(0)}}$ is the trivial $[1, 1]_{2}$ code.
\end{itemize}
\end{remark}

For a bent partition $\Gamma$ of $V_{n}$ of depth $p^{\frac{n}{2}}$ ($n \geq 4$) in class $\mathcal{WRBP}$, except for the case where $(p, n)=(2, 4)$ and there is $i_{0}$ such that $|A_{i_{0}}|=1$, we show that $\Gamma$ is obtained from a spread (after a shift).

\begin{theorem}\label{RevisedTheorem 3}
Let $n \geq 4$ be an even integer. Suppose $\Gamma=\{A_{i}, i \in V_{n/2}\}$ is a bent partition of $V_{n}$ in class $\mathcal{WRBP}$, where $p$ is odd, or $p=2$, $n \geq 6$, or $(p, n)=(2, 4)$ and $|A_{i}|\neq 1$ for all $i \in V_{n/2}$. Then there exist $\alpha \in V_{n}$ and $i_{0} \in V_{n/2}$ such that $A_{i}=(U_{i}-\alpha) \backslash \{-\alpha\}, i \neq i_{0}, A_{i_{0}}=(V \cup W)-\alpha$, where $\{V, W, U_{i}, i \in V_{n/2} \backslash \{i_{0}\}\}$ is a spread of $V_{n}$.
\end{theorem}

\begin{proof}
By Theorem \ref{Theorem 1} and its proof and Remark \ref{Remark 1}, there exist functions $G: V_{n} \rightarrow V_{n/2}$ and $h: V_{n} \rightarrow \mathbb{F}_{p}$ such that
\begin{equation}\label{rth31}
	\chi_{-a}(A_{i})=p^{n/2}\delta_{0}(a)+\epsilon(\delta_{G(a)}(i)p^{n/2}-1)\zeta_{p}^{h(a)}, a \in V_{n}, i \in V_{n/2},
\end{equation}
where $\epsilon \in \{\pm 1\}$ with $\epsilon=1$ if $p \neq 3$.
Furthermore, if $p$ is odd, then $h(0)=0, h(wa)=wh(a), G(wa)=G(a), w \in \mathbb{F}_{p}^{*}, a \in V_{n}$. We first show $\epsilon=1$. We only need to consider the case of $p=3$. To the contrary, assume $\epsilon=-1$. By Eq.\eqref{co5}, $D_{G, G(0)}=\{0\}$ and $|D_{G, i}|=3^{\frac{n}{2}}+1$ for any $i \neq G(0)$. Let $u \in V_{n/2} \backslash \{G(0)\}$ be fixed. Let $a_{0} \in V_{n}^{*}$ be with $G(a_{0})=u$. By Eq.\eqref{rth31},  $|A_{u}|=3^{\frac{n}{2}}+1$ and $\chi_{-a_{0}}(A_{u})=-(3^{\frac{n}{2}}-1)\zeta_{3}^{h(a_{0})}$. Denote $n_{j}=|\{x \in A_{u}: \langle -a_{0}, x \rangle_{n}=h(a_{0})+j\}, j=0, 1, 2$. Then  $(n_{0}+3^{\frac{n}{2}}-1)\zeta_{3}^{h(a_{0})}+n_{1}\zeta_{3}^{h(a_{0})+1}+n_{2}\zeta_{3}^{h(a_{0})+2}=0$, which implies that $n_{0}+3^{\frac{n}{2}}-1=n_{1}=n_{2}$. As $n_{0}+n_{1}+n_{2}=3^{\frac{n}{2}}+1$, we obtain $n_{0}=-3^{\frac{n}{2}-1}+1<0$ as $n \geq 4$, which is impossible. Thus, $\epsilon=1$. Note that when $p=2$, $h(0)=1$ if and only if $|A_{G(0)}|=1$. By Proposition \ref{Proposition 1}, $h(0)=0$ if $p=2$ and $n \geq 6$. Therefore, we obtain $h(0)=0$. By Eq.\eqref{rth31}, we have $|A_{G(0)}|=2p^{\frac{n}{2}}-1, |A_{i}|=p^{\frac{n}{2}}-1, i \neq G(0)$. Denote $i_{0}=G(0)$. For arbitrary fixed $a, b \in V_{n}$, we show that $h(a+b)=h(a)+h(b)$ by considering the following cases.
\begin{itemize}
  \item Case 1: $a=0$ or $b=0$. In this case, obviously $h(a+b)=h(a)+h(b)$ as $h(0)=0$.
  \item Case 2: $a, b \in V_{n}^{*}$ with $a+b=0$. In this case, $h(a+b)=h(a)+h(b)$ as $h(0)=0$ and $h(-a)=-h(a)$.
  \item Case 3: $a, b\in V_{n}^{*}$ with $a+b \neq 0$ and $(G(a), G(b)) \neq (i_{0}, i_{0})$. W.l.o.g, suppose $G(a) \neq i_{0}$. By Eq.\eqref{rth31}, we have
  \begin{equation*}
  \chi_{-a}(A_{G(a)})=(p^{\frac{n}{2}}-1)\zeta_{p}^{h(a)}.
  \end{equation*}
  As $|A_{G(a)}|=p^{\frac{n}{2}}-1$, we obtain $\langle -a, x \rangle_{n}=h(a)$ for all $x \in A_{G(a)}$. Then
  \begin{equation}\label{rth32}
  \chi_{-(a+b)}(A_{G(a)})=\sum_{x \in A_{G(a)}}\zeta_{p}^{\langle -a-b, x \rangle_{n}}=\sum_{x \in A_{G(a)}}\zeta_{p}^{h(a)-\langle b, x\rangle_{n}}=\zeta_{p}^{h(a)}\chi_{-b}(A_{G(a)}).
  \end{equation}
  By Eq.\eqref{rth31} and Eq.\eqref{rth32}, we obtain $h(a+b)=h(a)+h(b)$.
  \item Case 4: $a, b \in V_{n}^{*}$ with $a+b \neq 0$ and $G(a)=G(b)=i_{0}, G(a+b) \neq i_{0}$.
  By Eq. \eqref{rth31}, we have
  \begin{equation*}
  \chi_{-(a+b)}(A_{G(a+b)})=(p^{\frac{n}{2}}-1)\zeta_{p}^{h(a+b)}.
   \end{equation*}
  As $|A_{G(a+b)}|=p^{\frac{n}{2}}-1$, we obtain $\langle -a-b, x \rangle_{n}=h(a+b)$ for all $x \in A_{G(a+b)}$. From Eq. \eqref{rth31} and
  \begin{equation*}
	\chi_{-a}(A_{G(a+b)})=\sum_{x \in A_{G(a+b)}}\zeta_{p}^{\langle -a-b+b, x\rangle_{n}}=\sum_{x \in A_{G(a+b)}}\zeta_{p}^{h(a+b)+\langle b, x \rangle_{n}}=\zeta_{p}^{h(a+b)}\chi_{b}(A_{G(a+b)}),
\end{equation*}
  we obtain $h(a)=h(a+b)+h(-b)=h(a+b)-h(b)$, that is, $h(a+b)=h(a)+h(b)$.
  \item Case 5: $a, b \in V_{n}^{*}$ with $a+b \neq 0$ and $G(a)=G(b)=G(a+b)=i_{0}$. Notice that by Cases 1-4, for any $s, t \in V_{n}$ with $(G(s), G(t), G(s+t)) \neq (i_{0}, i_{0}, i_{0})$, we have $h(s+t)=h(s)+h(t)$. We first show that there exists $c \in V_{n}$ such that $G(c) \neq i_{0}$ and $G(a+c) \neq i_{0}$. To the contrary, assume that there is no $c \in V_{n}$ such that $G(c) \neq i_{0}$ and $G(a+c) \neq i_{0}$. Then if $c \in D_{G, V_{n/2} \backslash \{i_{0}\}}$, then $a+c \in D_{G, i_{0}}$. Thus,
      $|D_{G, V_{n/2} \backslash \{i_{0}\}}| \leq |D_{G, i_{0}}|$. By Eq.\eqref{co5} and Remark \ref{Remark 1}, we have $|D_{G, i_{0}}|=2p^{\frac{n}{2}}-1, |D_{G, i}|=p^{\frac{n}{2}}-1, i \neq i_{0}$. Then $(p^{\frac{n}{2}}-1)^{2} \leq 2p^{\frac{n}{2}}-1$, which is impossible as $n \geq 4$. Hence, there is $c \in V_{n}$ with $G(-c)=G(c) \neq i_{0}$ and $G(a+c) \neq i_{0}$. Then
      \begin{equation*}
      \begin{split}
      h(a+b)&=h\big((a+c)+(-c+b)\big)=h(a+c)+h(-c+b)\\
      &=h(a)+h(c)+h(-c)+h(b)\\
      &=h(a)+h(c)-h(c)+h(b)=h(a)+h(b).
      \end{split}
      \end{equation*}
\end{itemize}
Therefore, $h$ is a linear function from $V_{n}$ to $\mathbb{F}_{p}$, and then there exists $\alpha \in V_{n}$ such that $h(x)=\langle \alpha, x \rangle_{n}$. Next we show that $-\alpha \in A_{i_{0}}$. For any $i_{1} \in V_{n/2} \backslash \{i_{0}\}$, by Eq. \eqref{rth31}, 
\begin{equation}\label{rth33}
|A_{i_{1}}+\alpha|=p^{\frac{n}{2}}-1,\chi_{-a}(A_{i_{1}}+\alpha)=p^{\frac{n}{2}}\delta_{G(a)}(i_{1})-1, a \in V_{n}^{*}.
\end{equation}
By $|D_{G, i_{1}}^{*}|=p^{\frac{n}{2}}-1$ and Eq. \eqref{rth33}, we have
\begin{equation*}
p^{n}\delta_{A_{i_{1}}+\alpha}(0)=\sum_{x \in A_{i_{1}}+\alpha}\sum_{a \in V_{n}}\zeta_{p}^{\langle -a, x \rangle_{n}}=\sum_{a \in V_{n}}\chi_{-a}(A_{i_{1}}+\alpha)=p^{\frac{n}{2}}-1+(p^{\frac{n}{2}}-1)^{2}-(p^{n}-p^{\frac{n}{2}})=0,
\end{equation*}
which implies that $0 \notin A_{i_{1}}+\alpha$ for any $i_{1} \neq i_{0}$. Therefore, $0 \in A_{i_{0}}+\alpha$, that is, $-\alpha \in A_{i_{0}}$. 

For each $i \in V_{n/2} \backslash \{i_{0}\}$, let $U_{i}=(A_{i}+\alpha) \cup \{0\}$. Then for any $i \in V_{n/2} \backslash \{i_{0}\}$, $|U_{i}|=p^{\frac{n}{2}}$, and by Eq. \eqref{rth31}, $\chi_{-a}(U_{i})=0 \text{ or } p^{\frac{n}{2}}$ for any $a \in V_{n}^{*}$. Since $|U_{i}|=p^{\frac{n}{2}}$, then $\chi_{-a}(U_{i})=p^{\frac{n}{2}}$ if and only if $a \in U_{i}^{\bot}$, where for any set $S\subseteq V_{n}$, $S^{\bot}=\{x \in V_{n}: \langle x, y \rangle_{n}=0 \text{ for all } y \in S\}$ is a linear subspace. Let $<S>$ denote the linear subspace spanned by $S$. Note that $<S>^{\bot}=S^{\bot}$. By
   \begin{equation*}
   p^{\frac{n}{2}} |U_{i}^{\bot}|=\sum_{a \in V_{n}}\chi_{-a}(U_{i})=\sum_{x \in U_{i}}\sum_{a \in V_{n}}\zeta_{p}^{-\langle a, x\rangle_{n}}=p^n,
  \end{equation*}
   we have that the dimension of $U_{i}^{\bot}$ is $\frac{n}{2}$. Thus, the dimension of $<U_{i}>$ is $\frac{n}{2}$, and by $|U_{i}|=p^{\frac{n}{2}}$, $U_{i}$ is an $\frac{n}{2}$-dimensional linear subspace. For any $i \neq j \in V_{n/2}\backslash \{i_{0}\}$, we have
  $U_{i} \cap U_{j}=\{0\}$ as $A_{i} \cap A_{j}=\emptyset$. Hence, $\{U_{i}, i \in V_{n/2} \backslash \{i_{0}\}\}$ is a partial spread of $V_{n}$. By \cite[Theorem 3.1]{B1975P}, 
  $\{U_{i}, i \in V_{n/2} \backslash \{i_{0}\}\}$ can be extended to a spread $\{V, W, U_{i}, i \in V_{n/2} \backslash \{i_{0}\}\}$. Then $A_{i}=(U_{i}-\alpha) \backslash \{-\alpha\}, i \in V_{n/2} \backslash \{i_{0}\}$ and $A_{i_{0}}=(V \cup W)-\alpha$.
\end{proof}

By Theorem \ref{RevisedTheorem 3}, we show that all normal bent partitions of $V_{n}$ of depth $p^{\frac{n}{2}}$ in class $\mathcal{WRBP}$ are from spreads, where $n \geq 4$.

\begin{corollary}\label{Corollary 2}
Let $n \geq 4$ be an even integer. Let $\Omega=\{U, A_{i}, i \in V_{n/2}\}$ be a normal bent partition of $V_{n}$ in class $\mathcal{WRBP}$, where $U$ is an $n/2$-dimensional affine subspace of $V_{n}$. Then there exists $\gamma \in U$ such that $A_{i}=(U_{i}+\gamma) \backslash \{\gamma\}, i \in V_{n/2}$, where $\{U-\gamma, U_{i}, i \in V_{n/2}\}$ is a spread of $V_{n}$.
\end{corollary}

\begin{proof}
Since $\Omega$ is a normal bent partition in class $\mathcal{WRBP}$, then $|A_{i}|=p^{\frac{n}{2}}-1$ \cite{AM2022Be}. Let $\Gamma_{1}=\{U \cup A_{i_{1}}, A_{i}, i \in V_{n/2} \backslash \{i_{1}\}\}, \Gamma_{2}=\{U \cup A_{i_{2}}, A_{i}, i \in V_{n/2} \backslash \{i_{2}\}\}$, where $i_{1} \neq i_{2} \in V_{n/2}$. Then $\Gamma_{1}$ and $\Gamma_{2}$ are both bent partitions satisfying the conditions in Theorem \ref{RevisedTheorem 3}. For $\Gamma_{1}$, by Theorem \ref{RevisedTheorem 3} and its proof, there is $G: V_{n} \rightarrow V_{n/2}$ and linear function $h(x)=\langle \alpha, x\rangle_{n}$ (where $-\alpha \in U \cup A_{i_{1}}$) such that for $i \neq i_{1}, i_{2}$,
\begin{equation}\label{co21}
\chi_{-a}(A_{i})=p^{n/2}\delta_{0}(a)+(\delta_{G(a)}(i)p^{n/2}-1)\zeta_{p}^{h(a)}, a \in V_{n}.
\end{equation}
Also for $\Gamma_{2}$, there is $G': V_{n} \rightarrow V_{n/2}$ and linear function $h'(x)=\langle \alpha', x\rangle_{n}$ (where $-\alpha' \in U \cup A_{i_{2}}$) such that for $i \neq i_{1}, i_{2}$,
\begin{equation}\label{co22}
	\chi_{-a}(A_{i})=p^{n/2}\delta_{0}(a)+(\delta_{G'(a)}(i)p^{n/2}-1)\zeta_{p}^{h'(a)}, a \in V_{n}.
\end{equation}
By Eqs. \eqref{co21} and \eqref{co22}, we have $h=h'$ and then $-\alpha=-\alpha' \in (U \cup A_{i_{1}}) \cap (U \cup A_{i_{2}})=U$. Denote $\gamma=-\alpha$. Then $U-\gamma$ is an $\frac{n}{2}$-dimensional linear subspace. For $\Gamma_{1}$, by Theorem \ref{RevisedTheorem 3} and its proof, $U_{i} \triangleq (A_{i}-\gamma) \cup \{0\}, i \neq i_{1},$ are all $\frac{n}{2}$-dimensional linear subspaces. For $\Gamma_{2}$, by Theorem \ref{RevisedTheorem 3} and its proof, $U_{i_{1}} \triangleq (A_{i_{1}}-\gamma) \cup \{0\}$ is an $\frac{n}{2}$-dimensional linear subspace. Obviously, $U_{i} \cap U_{j}=\{0\}, i \neq j, U_{i} \cap (U-\gamma)=\{0\}$. Thus, $\{U-\gamma, U_{i}, i \in V_{n/2}\}$ is a spread.
\end{proof}

\section{Macwilliams-like identities for bent partitions in class $\mathcal{WRBP}$}
In this section, we give MacWilliams-like identities for bent partitions in class $\mathcal{WRBP}$. By using the obtained MacWilliams-like identities, some nonexistence results on bent partitions in class $\mathcal{WRBP}$ will be presented.

\begin{theorem}\label{Theorem 4}
For any $C \subseteq \mathbb{F}_{p}^{n}$, denote $W_{C}(x, y)\triangleq \sum_{a \in C}x^{n-wt(a)}y^{wt(a)}$. Let $\Gamma=\{A_{i}, i \in \mathbb{F}_{p}^{m}\}$ be a bent partition of $\mathbb{F}_{p}^{n}$ in class $\mathcal{WRBP}$, where $n$ is even, $m \leq \frac{n}{2}$, and $m\geq 2$ if $p=2$. Then there exist functions $G: \mathbb{F}_{p}^{n} \rightarrow \mathbb{F}_{p}^{m}$ and $h: \mathbb{F}_{p}^{n} \rightarrow \mathbb{F}_{p}$ such that for each $i \in \mathbb{F}_{p}^{m}$,
\begin{equation}\label{th40}
	\begin{split}
		&W_{A_{i}}(x+(p-1)y, x-y)\\
		&=p^{n-m}x^{n}+\epsilon p^{\frac{n}{2}-m}(p-1)^{-1}(x+(p-1)y)^{n}+\epsilon p^{\frac{n}{2}+1}(p-1)^{-1}W_{D_{G, i}\cap D_{h, 0}}(x, y)\\
		& \ \ \ -\epsilon p^{\frac{n}{2}}(p-1)^{-1}W_{D_{G, i}}(x, y)-\epsilon p^{\frac{n}{2}-m+1}(p-1)^{-1}W_{D_{h, 0}}(x, y),
	\end{split}
\end{equation}
where $\epsilon \in \{\pm 1\}$ with $\epsilon=1$ if $p \neq 3$.
\end{theorem}

\begin{proof}
Let $i \in \mathbb{F}_{p}^{m}$ be fixed. By Theorem \ref{Theorem 1} and its proof and Remark \ref{Remark 1}, for any $a \in \mathbb{F}_{p}^{n}$,
\begin{equation*}
\chi_{-a}(A_{i})=p^{n-m}\delta_{0}(a)+\epsilon p^{\frac{n}{2}-m}\zeta_{p}^{h(a)}(p^m\delta_{G(a)}(i)-1),
\end{equation*}
for some functions $G: \mathbb{F}_{p}^{n} \rightarrow \mathbb{F}_{p}^{m}$ and $h: \mathbb{F}_{p}^{n} \rightarrow \mathbb{F}_{p}$ with $G(ca)=G(a), h(ca)=ch(a), c \in \mathbb{F}_{p}^{*}$, where $\epsilon \in \{\pm 1\}$ with $\epsilon=1$ if $p \neq 3$. Then we have
\begin{equation*}
\begin{split}
&\sum_{a \in {\mathbb{F}_{p}^{n}}^{*}}\chi_{a}(A_{i})z^{wt(a)}\\
&=\sum_{a \in {\mathbb{F}_{p}^{n}}^{*}}\chi_{-a}(A_{i})z^{wt(-a)}\\
&=\sum_{a \in {\mathbb{F}_{p}^{n}}^{*}}\chi_{-a}(A_{i})z^{wt(a)}\\
&=\epsilon p^{\frac{n}{2}-m}\sum_{a \in {\mathbb{F}_{p}^{n}}^{*}}\zeta_{p}^{h(a)}(p^m \delta_{G(a)}(i)-1)z^{wt(a)}\\
&=\epsilon p^{\frac{n}{2}}\sum_{\substack{a \in {\mathbb{F}_{p}^{n}}^{*}\\G(a)=i}}\zeta_{p}^{h(a)}z^{wt(a)}-\epsilon p^{\frac{n}{2}-m}\sum_{a \in {\mathbb{F}_{p}^{n}}^{*}}\zeta_{p}^{h(a)}z^{wt(a)}\\
&=\epsilon p^{\frac{n}{2}}\sum_{\substack{\overline{a} \in {\mathbb{F}_{p}^{n}}^{*}/\mathbb{F}_{p}^{*}, c \in \mathbb{F}_{p}^{*}\\G(ca)=i}}\zeta_{p}^{h(ca)}z^{wt(ca)}-\epsilon p^{\frac{n}{2}-m}\sum_{\overline{a} \in {\mathbb{F}_{p}^{n}}^{*}/\mathbb{F}_{p}^{*}, c \in \mathbb{F}_{p}^{*}}\zeta_{p}^{h(ca)}z^{wt(ca)}\\
&=\epsilon p^{\frac{n}{2}}\sum_{\substack{\overline{a} \in {\mathbb{F}_{p}^{n}}^{*}/\mathbb{F}_{p}^{*}\\G(a)=i}}z^{wt(a)}\sum_{c \in \mathbb{F}_{p}^{*}}\zeta_{p}^{ch(a)}-\epsilon p^{\frac{n}{2}-m}\sum_{\overline{a} \in {\mathbb{F}_{p}^{n}}^{*}/\mathbb{F}_{p}^{*}}z^{wt(a)}\sum_{c \in \mathbb{F}_{p}^{*}}\zeta_{p}^{ch(a)}\\
&=\epsilon p^{\frac{n}{2}}\sum_{\substack{\overline{a} \in {\mathbb{F}_{p}^{n}}^{*}/\mathbb{F}_{p}^{*}\\G(a)=i}}z^{wt(a)} (p\delta_{0}(h(a))-1)-\epsilon p^{\frac{n}{2}-m}\sum_{\overline{a} \in {\mathbb{F}_{p}^{n}}^{*}/\mathbb{F}_{p}^{*}}z^{wt(a)}(p\delta_{0}(h(a))-1)\\
\end{split}
\end{equation*}
\begin{equation*}
\begin{split}
&=\epsilon p^{\frac{n}{2}+1}\sum_{\substack{\overline{a} \in {\mathbb{F}_{p}^{n}}^{*}/\mathbb{F}_{p}^{*}\\G(a)=i, h(a)=0}}z^{wt(a)}-\epsilon p^{\frac{n}{2}}\sum_{\substack{\overline{a} \in {\mathbb{F}_{p}^{n}}^{*}/\mathbb{F}_{p}^{*}\\G(a)=i}}z^{wt(a)}-\epsilon p^{\frac{n}{2}-m+1}\sum_{\substack{\overline{a} \in {\mathbb{F}_{p}^{n}}^{*}/\mathbb{F}_{p}^{*}\\h(a)=0}}z^{wt(a)}+\epsilon p^{\frac{n}{2}-m}\sum_{\overline{a} \in {\mathbb{F}_{p}^{n}}^{*}/\mathbb{F}_{p}^{*}}z^{wt(a)}\\
&=\epsilon p^{\frac{n}{2}+1}(p-1)^{-1}\sum_{\substack{a \in {\mathbb{F}_{p}^{n}}^{*}\\G(a)=i, h(a)=0}}z^{wt(a)}-\epsilon p^{\frac{n}{2}}(p-1)^{-1}\sum_{\substack{a \in {\mathbb{F}_{p}^{n}}^{*}\\G(a)=i}}z^{wt(a)}\\
& \ \ \ -\epsilon p^{\frac{n}{2}-m+1}(p-1)^{-1}\sum_{\substack{a \in {\mathbb{F}_{p}^{n}}^{*}\\h(a)=0}}z^{wt(a)}+\epsilon p^{\frac{n}{2}-m}(p-1)^{-1}\sum_{a \in {\mathbb{F}_{p}^{n}}^{*}}z^{wt(a)}\\
&=\epsilon p^{\frac{n}{2}+1}(p-1)^{-1}\sum_{a \in D_{G, i} \cap D_{h, 0}}z^{wt(a)}-\epsilon p^{\frac{n}{2}}(p-1)^{-1}\sum_{a \in D_{G, i}}z^{wt(a)}-\epsilon p^{\frac{n}{2}-m+1}(p-1)^{-1}\sum_{a \in D_{h, 0}}z^{wt(a)}\\
& \ \ \ +\epsilon p^{\frac{n}{2}-m}(p-1)^{-1}(1+(p-1)z)^{n}-\epsilon p^{\frac{n}{2}+1}(p-1)^{-1}\delta_{G(0)}(i)\delta_{h(0)}(0)+\epsilon p^{\frac{n}{2}}(p-1)^{-1}\delta_{G(0)}(i)\\
& \ \ \ +\epsilon p^{\frac{n}{2}-m+1}(p-1)^{-1}\delta_{h(0)}(0)-\epsilon p^{\frac{n}{2}-m}(p-1)^{-1}.
\end{split}
\end{equation*}
Note that $h(0)=0$ if $p$ is odd. Therefore,
\begin{equation}\label{th41}
\begin{split}
&\sum_{a \in \mathbb{F}_{p}^{n}}\chi_{a}(A_{i})z^{wt(a)}\\
&=|A_{i}|+\sum_{a \in {\mathbb{F}_{p}^{n}}^{*}}\chi_{a}(A_{i})z^{wt(a)}\\
&=p^{n-m}+\epsilon p^{\frac{n}{2}+1}(p-1)^{-1}\sum_{a \in D_{G, i} \cap D_{h, 0}}z^{wt(a)}-\epsilon p^{\frac{n}{2}}(p-1)^{-1}\sum_{a \in D_{G, i}}z^{wt(a)}
\\
& \ \ \ -\epsilon p^{\frac{n}{2}-m+1}(p-1)^{-1}\sum_{a \in D_{h, 0}}z^{wt(a)}+\epsilon p^{\frac{n}{2}-m}(p-1)^{-1}(1+(p-1)z)^{n}.
\end{split}
\end{equation}
Set $z=\frac{y}{x}$, by Eq. \eqref{th41}, we obtain
\begin{equation}\label{th42}
\begin{split}
&\sum_{a \in \mathbb{F}_{p}^{n}}\chi_{a}(A_{i})x^{n-wt(a)}y^{wt(a)}\\
&=p^{n-m}x^{n}+\epsilon p^{\frac{n}{2}-m}(p-1)^{-1}(x+(p-1)y)^{n}+\epsilon p^{\frac{n}{2}+1}(p-1)^{-1}W_{D_{G, i}\cap D_{h, 0}}(x, y)\\
& \ \ \ -\epsilon p^{\frac{n}{2}}(p-1)^{-1}W_{D_{G, i}}(x, y)-\epsilon p^{\frac{n}{2}-m+1}(p-1)^{-1}W_{D_{h, 0}}(x, y).
\end{split}
\end{equation}
By \cite[Lemma 2.1]{CHK2014Lo} (see also \cite[Lemma 3.1]{HLL2014Ne}), for any subset $C$ of $\mathbb{F}_{p}^{n}$,
\begin{equation}\label{th43}
W_{C}(x+(p-1)y, x-y)=\sum_{a \in \mathbb{F}_{p}^{n}}\chi_{a}(C)x^{n-wt(a)}y^{wt(a)}.
\end{equation}
Combining Eq. \eqref{th42} and Eq. \eqref{th43}, we have that Eq. \eqref{th40} holds.
\end{proof}

By Theorem \ref{Theorem 4}, we obtain the following corollary.

\begin{corollary}
Keep the same notations as in Theorem \ref{Theorem 4}. For any $C \subseteq \mathbb{F}_{p}^{n}$, denote $M_{j}(C)=|\{a \in C: wt(a)=j\}|$, where $0 \leq j \leq n$. For each $i \in \mathbb{F}_{p}^{m}$ and $1 \leq j \leq n$,
\begin{equation}\label{co0}
	\begin{split}
		M_{j}(A_{i})&=\binom{n}{j}p^{-m}(p-1)^{j}+\epsilon p^{-\frac{n}{2}+1}(p-1)^{-1}\sum_{t=0}^{n}M_{t}(D_{G, i} \cap D_{h, 0})K_{j}(t)\\
		& \ \ \ -\epsilon p^{-\frac{n}{2}}(p-1)^{-1}\sum_{t=0}^{n}M_{t}(D_{G, i})K_{j}(t)-\epsilon p^{-\frac{n}{2}-m+1}(p-1)^{-1}\sum_{t=0}^{n}M_{t}(D_{h, 0})K_{j}(t),
	\end{split}
\end{equation}
where $K_{j}(t)=\sum_{k=0}^{j}\binom{t}{k}\binom{n-t}{j-k}(-1)^{k}(p-1)^{j-k}$ is the Krawtchouk number.
\end{corollary}

\begin{proof}
 Set $X=x+(p-1)y, Y=x-y$. Then $x=\frac{1}{p}(X+(p-1)Y), y=\frac{1}{p}(X-Y)$. Let $i \in \mathbb{F}_{p}^{m}$ be fixed. By Eq. \eqref{th40}, we have
\begin{equation}\label{co1}
\begin{split}
W_{A_{i}}(X, Y)&=p^{-m}(X+(p-1)Y)^{n}+\epsilon p^{\frac{n}{2}-m}(p-1)^{-1}X^{n}\\
& \ \ \ +\epsilon p^{-\frac{n}{2}+1}(p-1)^{-1}W_{D_{G, i} \cap D_{h, 0}}(X+(p-1)Y, X-Y)\\
& \ \ \ -\epsilon p^{-\frac{n}{2}}(p-1)^{-1}W_{D_{G, i}}(X+(p-1)Y, X-Y)\\
& \ \ \ -\epsilon p^{-\frac{n}{2}-m+1}(p-1)^{-1}W_{D_{h, 0}}(X+(p-1)Y, X-Y).
\end{split}
\end{equation}
By setting $X=1$ in Eq. \eqref{co1}, we obtain
\begin{equation}\label{co2}
\begin{split}
&\sum_{j=0}^{n}M_{j}(A_{i})Y^{j}\\
&=p^{-m}(1+(p-1)Y)^{n}+\epsilon p^{\frac{n}{2}-m}(p-1)^{-1}\\
& \ \ \ +\epsilon p^{-\frac{n}{2}+1}(p-1)^{-1}\sum_{j=0}^{n}M_{j}(D_{G, i} \cap D_{h, 0})(1+(p-1)Y)^{n-j}(1-Y)^{j}\\
& \ \ \ -\epsilon p^{-\frac{n}{2}}(p-1)^{-1}\sum_{j=0}^{n}M_{j}(D_{G, i})(1+(p-1)Y)^{n-j}(1-Y)^{j}\\
& \ \ \ -\epsilon p^{-\frac{n}{2}-m+1}(p-1)^{-1}\sum_{j=0}^{n}M_{j}(D_{h, 0})(1+(p-1)Y)^{n-j}(1-Y)^{j}.
\end{split}
\end{equation}
For $1 \leq j \leq n$ and $0 \leq t \leq n$, the coefficient of $Y^{j}$ in the expansion of $(1+(p-1)Y)^{n-t}(1-Y)^{t}$ is the Krawtchouk number $K_{j}(t)=\sum_{k=0}^{j}\binom{t}{k}\binom{n-t}{j-k}(-1)^{k}(p-1)^{j-k}$. By comparing the coefficient of $Y^{j}$ on both sides of Eq. \eqref{co2}, we have that Eq. \eqref{co0} holds.
\end{proof}

When $p$ is odd, based on Theorem \ref{Theorem 4}, we present some nonexistence results on bent partitions in class $\mathcal{WRBP}$.
\begin{theorem} \label{Theorem 5}
Let $p$ be an odd prime, $k$ be a nonnegative integer, $m$ be a positive integer, and let $\epsilon \in \{\pm 1\}$ be a constant with $\epsilon=1$ if $p \neq 3$. When $\epsilon=1$, let $n_{k, p, m}$ be the smallest even positive integer satisfying
\begin{equation*}
\binom{\frac{n_{k, p, m}}{2}+k+1}{k+1}<\frac{p(p-2)(p^m-1)-1}{p^m(p-1)}\Big(\frac{p}{p-1}\Big)^{\frac{n_{k, p, m}}{2}}-(p-1),
\end{equation*}
and when $\epsilon=-1$ (in this case, $p=3$), let $n_{k, 3, m}$ be the smallest even positive integer satisfying
\begin{equation*}
\binom{\frac{n_{k, 3, m}}{2}+k+1}{k+1}<\frac{2 \cdot 3^m-5}{2 \cdot 3^m}\Big(\frac{3}{2}\Big)^{
\frac{n_{k, 3, m}}{2}}-2.
\end{equation*}
Let $n\geq n_{k, p, m}$ be an even positive integer. Then there is no bent partition $\Gamma=\{A_{i}, i \in \mathbb{F}_{p}^{m}\}$ of $\mathbb{F}_{p}^{n}$ in class $\mathcal{WRBP}$ (where the corresponding sign is $\epsilon$) with $\{x \in \mathbb{F}_{p}^{n}: 0\leq wt(x)\leq \frac{n}{2}-k-1\}\subseteq A_{i_{0}}$ for some $i_{0} \in \mathbb{F}_{p}^{m}$.
\end{theorem}
\begin{proof}
Let $\Gamma=\{A_{i}, i \in \mathbb{F}_{p}^{m}\}$ be a bent partition of $\mathbb{F}_{p}^{n}$ in class $\mathcal{WRBP}$. By Theorem \ref{Theorem 4}, there are functions $G: \mathbb{F}_{p}^{n} \rightarrow \mathbb{F}_{p}^{m}$ and $h: \mathbb{F}_{p}^{n} \rightarrow \mathbb{F}_{p}$ such that Eq. \eqref{th40} holds. Suppose $0 \in A_{i_{0}}$. Define $A=\bigcup_{i \in \mathbb{F}_{p}^{m} \backslash \{i_{0}\}}A_{i}, D=\bigcup_{i \in \mathbb{F}_{p}^{m} \backslash \{i_{0}\}}D_{G, i}, S=D_{h, 0}$. By Eq. \eqref{th40}, we have
\begin{equation}\label{th51}
	\begin{split}
	& W_{A}(x+(p-1)y, x-y)\\
	&=\sum_{i \in \mathbb{F}_{p}^{m} \backslash \{i_{0}\}}W_{A_{i}}(x+(p-1)y, x-y)\\
	&=p^{n-m}(p^m-1)x^{n}+\epsilon p^{\frac{n}{2}-m}\frac{p^m-1}{p-1}(x+(p-1)y)^{n}+\epsilon p^{\frac{n}{2}+1}(p-1)^{-1}W_{D \cap S}(x, y)\\
	& \ \ \ -\epsilon p^{\frac{n}{2}}(p-1)^{-1}W_{D}(x, y)-\epsilon p^{\frac{n}{2}-m+1}\frac{p^m-1}{p-1}W_{S}(x, y).
	\end{split}
\end{equation}
Denote $M_{j}(A)=|\{x \in A: wt(x)=j\}|, 0 \leq j \leq n$. By the proof of \cite[Lemma 5.2]{HLL2014Ne}, 
for each $0 \leq k \leq \frac{n}{2}-1$,
\begin{equation*}
	\sum_{j=\frac{n}{2}-k}^{\frac{n}{2}}M_{j}(A)\binom{n-j}{\frac{n}{2}-1}\leq (p-1)^{\frac{n}{2}}\binom{n}{\frac{n}{2}-1}\binom{\frac{n}{2}+k+1}{k+1}.
\end{equation*}
If $\{x
\in \mathbb{F}_{p}^{n}: 0 \leq wt(x) \leq \frac{n}{2}-k-1\} \subseteq A_{i_{0}}$, that is, $M_{j}(A)=0$ for any $0 \leq j \leq \frac{n}{2}-k-1$, then
\begin{equation}\label{th55}
	\begin{split}
		\sum_{j=0}^{\frac{n}{2}+1}M_{j}(A)\binom{n-j}{\frac{n}{2}-1}&=\sum_{j=\frac{n}{2}-k}^{\frac{n}{2}}M_{j}(A)\binom{n-j}{\frac{n}{2}-1}+M_{\frac{n}{2}+1}(A)\\
		& \leq (p-1)^{\frac{n}{2}}\binom{n}{\frac{n}{2}-1}\binom{\frac{n}{2}+k+1}{k+1} +(p-1)^{\frac{n}{2}+1}\binom{n}{\frac{n}{2}-1}.
	\end{split}
\end{equation}
For any set $C \subseteq \mathbb{F}_{p}^{n}$, let $\overline{C}$ denote its complement, that is, $\overline{C}=\{x \in \mathbb{F}_{p}^{n}: x \notin C\}$.

Case I: $\epsilon=1$. Set $x=\frac{1}{p}+z, y=\frac{1}{p}$ in Eq. \eqref{th51}, then we obtain
\begin{equation}\label{th52}
\begin{split}
	&\sum_{j=0}^{n}M_{j}(A)(1+z)^{n-j}z^{j}\\
	&=p^{-m}(p^m-1)(1+pz)^{n}+p^{\frac{n}{2}-m}\frac{p^m-1}{p-1}(1+z)^{n}+p^{\frac{n}{2}+1}(p-1)^{-1}W_{D \cap S}(p^{-1}+z, p^{-1})\\
	& \ \ \ -p^{\frac{n}{2}}(p-1)^{-1}\Big((1+z)^{n}-W_{\overline{D}}(p^{-1}+z, p^{-1})\Big)-p^{\frac{n}{2}-m+1}\frac{p^m-1}{p-1}\Big((1+z)^{n}-W_{\overline{S}}(p^{-1}+z, p^{-1})\Big).
\end{split}
\end{equation}
Comparing the coefficients of $z^{s} \ (0 \leq s \leq n)$ on both sides of Eq. \eqref{th52}, we deduce that
\begin{equation}\label{th53}
	\begin{split}
	&\sum_{j=0}^{s}M_{j}(A)\binom{n-j}{n-s} \\
	& \geq p^{s-m}(p^m-1)\binom{n}{s}+p^{\frac{n}{2}-m}\frac{p^m-1}{p-1}\binom{n}{s}-p^{\frac{n}{2}}(p-1)^{-1}\binom{n}{s}-p^{\frac{n}{2}-m+1}\frac{p^m-1}{p-1}\binom{n}{s}\\
	&=(p^m-1)(p^{s-m}-p^{\frac{n}{2}-m})\binom{n}{s}-p^{\frac{n}{2}}(p-1)^{-1}\binom{n}{s}.
	\end{split}
\end{equation}
Set $s=\frac{n}{2}+1$ in Ineq.\eqref{th53}, then
\begin{equation}\label{th54}
\sum_{j=0}^{\frac{n}{2}+1}M_{j}(A)\binom{n-j}{\frac{n}{2}-1}\geq (p^m-1)(p-1)p^{\frac{n}{2}-m}\binom{n}{\frac{n}{2}-1}-p^{\frac{n}{2}}(p-1)^{-1}\binom{n}{\frac{n}{2}-1}.
\end{equation}
Comparing Ineqs. \eqref{th55} and \eqref{th54}, we have
\begin{equation*}
	\binom{\frac{n}{2}+k+1}{k+1}\geq \frac{p(p-2)(p^m-1)-1}{p^m(p-1)}\Big(\frac{p}{p-1}\Big)^{\frac{n}{2}}-(p-1).
\end{equation*}

Case II: $\epsilon=-1$. In this case, $p=3$. Set $x=\frac{1}{3}+z, y=\frac{1}{3}$ in Eq. \eqref{th51}, then 
\begin{equation}\label{th56}
	\begin{split}
		&\sum_{j=0}^{n}M_{j}(A)(1+z)^{n-j}z^{j}\\
		&=3^{-m}(3^m-1)(1+3z)^{n}-3^{\frac{n}{2}-m}\frac{3^m-1}{2}(1+z)^{n}-3^{\frac{n}{2}+1}2^{-1}\Big((1+z)^{n}-W_{\overline{D \cap S}}(3^{-1}+z, 3^{-1})\Big)\\
		& \ \ \ +3^{\frac{n}{2}}2^{-1}W_{D}(3^{-1}+z, 3^{-1})+3^{\frac{n}{2}-m+1}\frac{3^m-1}{2}W_{S}(3^{-1}+z, 3^{-1}).
	\end{split}
\end{equation}
Comparing the coefficients of $z^{s} \ (0 \leq s \leq n)$ on both sides of Eq. \eqref{th56}, we deduce that
\begin{equation}\label{th57}
		\sum_{j=0}^{s}M_{j}(A)\binom{n-j}{n-s}\geq 3^{s-m}(3^m-1)\binom{n}{s}-3^{\frac{n}{2}-m}\frac{3^m-1}{2}\binom{n}{s}-3^{\frac{n}{2}+1}2^{-1}\binom{n}{s}.
\end{equation}
Set $s=\frac{n}{2}+1$ in Ineq. \eqref{th57}, then
\begin{equation}\label{th58}
	\sum_{j=0}^{\frac{n}{2}+1}M_{j}(A)\binom{n-j}{\frac{n}{2}-1}\geq \Big(\frac{5}{2}(3^m-1)3^{\frac{n}{2}-m}-3^{\frac{n}{2}+1}2^{-1}\Big)\binom{n}{\frac{n}{2}-1}.
\end{equation}
Comparing Ineqs. \eqref{th55} and \eqref{th58}, we have
\begin{equation*}
	\binom{\frac{n}{2}+k+1}{k+1} \geq \frac{2\cdot 3^m-5}{2 \cdot 3^m}\Big(\frac{3}{2}\Big)^{\frac{n}{2}}-2.
\end{equation*}
\end{proof}

\begin{example}
Let $k=0, p=3, m=2, n=14$. It is easy to check that $n$ satisfies the condition in Theorem \ref{Theorem 5} (whenever $\epsilon=1$ or $-1$). By Theorem \ref{Theorem 5}, there is no bent partition $\Gamma=\{A_{i}: i \in \mathbb{F}_{3}^{2}\}$ of $\mathbb{F}_{3}^{14}$ in class $\mathcal{WRBP}$ with $\{x \in \mathbb{F}_{3}^{14}: 0 \leq wt(x) \leq 6\} \subseteq A_{i_{0}}$ for some $i_{0} \in \mathbb{F}_{3}^{2}$. Notice that $|\{x \in \mathbb{F}_{3}^{14}: 0 \leq wt(x) \leq 6\}|=275577$, and by Eq. \eqref{add}, for each bent partition $\Gamma=\{A_{i}: i \in \mathbb{F}_{3}^{2}\}$ of $\mathbb{F}_{3}^{14}$ in class $\mathcal{WRBP}$, $|A_{i}| \geq 529497$.
\end{example}

\section{Conclusion}
Let $\Gamma=\{A_{i}, i \in V_{m}\}$ be a weakly regular bent partition of $V_{n}$. In this paper, we gave characterizations of $\Gamma$ in terms of the duals of component functions of vectorial bent functions, Buston Hadamard matrices, and two classes of linear codes. We proved that all bent functions $f$ generated by $\Gamma$ are regular if $p \neq 3$, that is, $\epsilon=1$ if $p \neq 3$. When $n \geq 4, m=\frac{n}{2}$ (that is, the depth of a weakly regular bent partition reaches its maximal value), except for the case where $(p, n)=(2, 4)$ and $|A_{i_{0}}|=1$ for some $i_{0}$, we showed that the bent partition $\Gamma$ is constructed by using a spread (after a shift). In particular, each weakly regular normal bent partition is constructed by using a spread (after a shift), where $n \geq 4$. Furthermore, we gave Macwilliams-like identities for the bent partition $\Gamma$. When $p$ is odd, some nonexistence results on weakly regular bent partitions were presented.

\end{document}